\documentclass[11pt,reqno,twoside]{amsart}

\usepackage{amssymb}
\usepackage{epstopdf}
\usepackage[asymmetric,top=3.5cm,bottom=3.3cm,left=3.1cm,right=3.1cm]{geometry}
\usepackage{booktabs} 
\usepackage{array} 
\usepackage{paralist} 
\usepackage{verbatim} 
\usepackage{graphicx,subfigure,threeparttable}

\usepackage{tabularx}
\usepackage{amsmath,amsfonts,amsthm,mathrsfs,cite}
\usepackage[usenames]{color}
\usepackage{bm}
\usepackage{tabularx}
\usepackage{longtable}
\usepackage{booktabs} 

\newtheorem{thm}{Theorem}[section]

\newtheorem{lem}{Lemma}[section]

\theoremstyle{definition}
\newtheorem{defn}{Definition}[section]
\theoremstyle{remark}
\newtheorem{example}{Example}

\newtheorem{rem}{Remark}[section]
\numberwithin{equation}{section}

\allowdisplaybreaks

\usepackage{hyperref}
\usepackage{xcolor}
\hypersetup{
  colorlinks,
  linkcolor={blue!80!black},
  urlcolor={blue!80!black},
  citecolor={blue!80!black}
}
\usepackage[capitalize,noabbrev]{cleveref}

\title[Quantitative Properties of EM transmission eigenfunctions and application]{Geometrical and topological properties of transmission resonance and artificial mirage}

\author{Youjun Deng}
\address{School of Mathematics and Statistics, HNP-LAMA, Central South University, Changsha, Hunan, China
}
\email{youjundeng@csu.edu.cn}

\author{Hongyu Liu}
\address{Department of Mathematics, City University of Hong Kong, Kowloon, Hong Kong, China}
\email{hongyu.liuip@gmail.com, hongyliu@cityu.edu.hk}

\author{Xianchao Wang}
\address{Department of Mathematics, City University of Hong Kong, Kowloon, Hong Kong, China}
\email{xcwang90@gmail.com}

\author{Wei Wu}
\address{Department of Mathematics, Hong Kong Baptist University, Hong Kong, China}
\email{wei-wu@hkbu.edu.hk}

\date{} 
\begin{document}

\maketitle

\begin{abstract}
Transmission eigenfunctions are certain interior resonant modes that are of central importance to the wave scattering theory. In this paper, we present the discovery of novel global rigidity properties of the transmission eigenfunctions associated with the Maxwell system. It is shown that the transmission eigenfunctions carry the geometrical and topological information of the underlying domain. We present both analytical and numerical results of these intriguing rigidity properties. As an interesting application, we propose an illusion scheme of artificially generating a mirage image of any given optical object.

\medskip



\end{abstract}

\section{Introduction}


We are concerned with the spectral geometry of the interior transmission eigenvalue problem arising in the time-harmonic electromagnetic (EM) scattering described by the Maxwell system. Let $\Omega$ be a bounded Lipschitz domain in $\mathbb{R}^3$ such that $\mathbb{R}^3\backslash\overline{\Omega}$ is connected. Let $\varepsilon$ and $\mu$ be bounded positive functions such that $\mathrm{supp}(1-\varepsilon)\cup\mathrm{supp}(1-\mu)\subset\Omega$. In the physical context, $\varepsilon$ and $\mu$ are the optical parameters of the space medium, and respectively signify the electric permittivity and magnetic permeability. Throughout, we assume that $\mu\equiv 1$. Consider a pair of incident EM waves $(\mathbf{E}^i, \mathbf{H}^i)$ that are entire solutions to the following Maxwell system:
\begin{equation}\label{eq:in1}
\nabla\wedge\mathbf{E}^i-\mathrm{i}k\mathbf{H}^i=\mathbf{0},\quad\nabla\wedge\mathbf{H}^i+\mathrm{i}k\mathbf{E}^i=\mathbf{0}\quad\mbox{in}\ \ \mathbb{R}^3,
\end{equation}
where $k\in\mathbb{R}_+$ and $\mathrm{i}:=\sqrt{-1}$. Here, $k$ signifies the wavenumber of the EM wave propagation. The EM inhomogeneity $(\Omega, \varepsilon)$ interrupts the incident fields and leads to the scattered fields $\mathbf{E}^s$ and $\mathbf{H}^s$. Let $\mathbf{E}(\mathbf{x})=\mathbf{E}^i(\mathbf{x})+\mathbf{E}^s(\mathbf{x})$ and $\mathbf{H}(\mathbf{x})=\mathbf{H}^i(\mathbf{x})+\mathbf{H}^s(\mathbf{x})$ denote the total electric and magnetic fields, respectively. The EM scattering is governed by the following Maxwell system:
\begin{equation}\label{eq:maxwell1}
\left\{
\begin{aligned}
&~~\nabla\wedge\mathbf{E}-\mathrm{i}k\mathbf{H}=\mathbf{0}, && \mbox{in}\ \mathbb{R}^3,\\
&~~\nabla\wedge\mathbf{H}+\mathrm{i}k\varepsilon\mathbf{E}=\mathbf{0} && \mbox{in}\ \mathbb{R}^3,\\
&\lim _{|{\mathbf x}| \rightarrow \infty}\left(\mathbf{H}^{s} \wedge {\mathbf x}-|{\mathbf x} | \mathbf{E}^{s}\right)={\mathbf 0}.
\end{aligned}
\right.
\end{equation}
The last limit in \eqref{eq:maxwell1} is known as the Silver-M\"uller radiation condition, which characterizes the outgoing nature of the scattered fields and holds uniformly in the angular variable $\hat{\mathbf{x}}:=\mathbf{x}/|\mathbf{x}|\in\mathbb{S}^{2}:=\left\{\mathbf{x} \in \mathbb{R}^{3} ;|\mathbf{x}|=1\right\}$. The well-posedness of the scattering problem \eqref{eq:maxwell1} can be conveniently found in \cite{LRX,Monk}. There exists a unique pair of solutions $({\mathbf E}, {\mathbf H}) \in$ $H_{\mathrm{loc}}(\mathrm{curl} , \mathbb{R}^{3}) \times H_{\mathrm{loc }}(\mathrm{curl}, \mathbb{R}^{3})$ which admits the following asymptotic expansions as $|\mathbf{x}|\rightarrow\infty$:
\begin{equation}\label{eq:far}
	\mathbf{Q}(\mathbf{x} )=\mathbf{Q}^i(\mathbf{x})+\frac{{e}^{\mathrm{i} k |\mathbf{x}|}}{|\mathbf{x} |}\mathbf{Q}_{\infty}(\hat{\mathbf{x}} )+\mathcal{O}\left(\frac{1}{|\mathbf{x} |^2} \right),\quad \mathbf{Q}:=\mathbf{E}, \mathbf{H}.
\end{equation}
The functions $\mathbf{E}_{\infty}(\hat{\mathbf{x} })$ and $\mathbf{H}_{\infty}(\hat{\mathbf{x} })$ in \eqref{eq:far} are respectively referred to as the electric and magnetic far field patterns, and satisfy the following one-to-one correspondence,
\[
\mathbf{E}_\infty(\hat{\mathbf{x}})=-\hat{\mathbf{x}}\wedge\mathbf{H}_\infty(\hat{\mathbf{x}})\quad\mbox{and}\quad \mathbf{H}_\infty(\mathbf{x})=\hat{\mathbf{x}}\wedge\mathbf{E}_\infty(\hat{\mathbf{x}}),\quad\forall\hat{\mathbf{x}}\in\mathbb{S}^2.
\]

We next consider a particular case that non-scattering, a.k.a invisibility, occurs, namely $\mathbf{Q}_\infty\equiv\mathbf{0}$. In such a case, by Rellich's Theorem \cite{CK}, one has $\mathbf{Q}^s=\mathbf{0}$ in $\mathbf{R}^3\backslash\overline{\Omega}$. With such an observation, one can directly verify that $\mathbf{Q}_1=\mathbf{Q}|_{\Omega}$ and $\mathbf{Q}_2=\mathbf{Q}^i|_{\Omega}$ fulfil the following PDE system:
\begin{equation}\label{eq:tcase1}
\left\{
\begin{aligned}
 &\nabla\wedge\mathbf{E}_1-\mathrm{i}k\mathbf{H}_1=\mathbf{0},&& \nabla\wedge\mathbf{H}_1+\mathrm{i}k\varepsilon \mathbf{E}_1=\mathbf{0}\quad && \mbox{in}\ \ \Omega,\medskip\\
 &\nabla\wedge\mathbf{E}_2-\mathrm{i}k\mathbf{H}_2=\mathbf{0},&& \nabla\wedge\mathbf{H}_2+\mathrm{i}k \mathbf{E}_2=\mathbf{0}\quad && \mbox{in}\ \ \Omega,\medskip\\
 &\nu\wedge\mathbf{E}_1=\nu\wedge\mathbf{E}_2,&& \nu\wedge\mathbf{H}_1=\nu\wedge\mathbf{H}_2\quad && \mbox{on}\ \ \partial\Omega,
\end{aligned}
\right.
\end{equation}
where and also in what follows, $\nu$ signifies the exterior unit normal vector to $\partial\Omega$. The system \eqref{eq:tcase1} is referred to as the Maxwell transmission eigenvalue problem. It is clear that $\mathbf{E}_j=\mathbf{0}$ and $\mathbf{H}_j=\mathbf{0}$, $j=1,2$, are trivial solutions. If there exist nontrivial solutions, $(\mathbf{E}_j,\mathbf{H}_j)_{j=1,2}$ are called the transmission eigenfunctions and $k$ is the associated transmission eigenvalue. In this paper, we shall be mainly concerned with real transmission eigenvalues. Hence, if no scattering/invisibility occurs, then the restrictions of the incident and total EM fields are transmission eigenfunctions and the wavenumber is an eigenvalue. After eliminating $\mathbf{H}_j$, $j=1, 2$, in \eqref{eq:tcase1}, we have the following reduced formulation of the transmission eigenvalue problem for $\mathbf{E}_j\in H(\mathrm{curl}, \Omega)$, $j=1, 2$:
\begin{equation}\label{eq:tcase12}
\left\{
\begin{aligned}
&\nabla\wedge(\nabla\wedge\mathbf{E}_1)-k^2\varepsilon\mathbf{E}_1=\mathbf{0},&& \nabla\cdot(\varepsilon\mathbf{E}_1)=\mathbf{0}&&\mbox{in}\ \ \Omega,\medskip\\
&\nabla\wedge(\nabla\wedge\mathbf{E}_2)-k^2\mathbf{E}_2=\mathbf{0},&& \nabla\cdot\mathbf{E}_2=\mathbf{0}&&\mbox{in}\ \ \Omega,\medskip\\
&\nu\wedge\mathbf{E}_1=\nu\wedge\mathbf{E}_2, && \nu\wedge(\nabla\wedge \mathbf{E}_1)=\nu\wedge(\nabla\wedge\mathbf{E}_2)&& \mbox{on}\ \ \partial\Omega.
\end{aligned}
\right.
\end{equation}
For a convenient reference, if one considers the transverse EM scattering (cf. \cite{CDL}), one can deduce the following transmission eigenvalue problem associated with the Helmholtz equation in $\mathbb{R}^2$:
\begin{equation}\label{eq:trans1}
\Delta w+k^2\mathbf{n}^2w=0,\ \
\Delta v+k^2 v =0\ \ \text{in} \ \ \Omega;\ \ \
w=v,\ \ \partial_\nu w=\partial_\nu v  \ \text{on}\ \ \partial \Omega,
\end{equation}
where $\mathbf{n}:=\sqrt{\epsilon}$ is known as the refractive index of the inhomogeneous medium, and $w, v$ correspond to the transverse parts of $\mathbf{E}$ and $\mathbf{H}$.

The spectral study of the transmission eigenvalue problems has a long and colourful history in the literature and has become a central topic in the inverse scattering theory.
The spectral properties of transmission eigenvalues have been intensively and extensively investigated in the literature, including the existence, infiniteness and Weyl's law, and we refer to \cite{CCH,CHreview,CKreview} and the references cited in for the state-of-the-art development on this aspect. Recently, the spectral geometry of transmission eigenfunctions has received considerable attentions in the literature in different physical contexts \cite{Bsource,BLLW,EBL,BL2016,BL2018,BL2017b,BXL,CV1,CLW1,CDHLW,CX,DDL,DCL}, and we refer \cite{Liureview} for survey and review on this aspect. It is pointed out that most of the aforementioned works are concerned with local geometric properties of the transmission eigenfunctions except \cite{CDHLW,Deng2021}, where global rigidity properties of the transmission eigenfunctions were discovered. Our study in the current article is strongly motivated by the related study in \cite{CDHLW}, where it is established that the transmission eigenfunctions associated with the Helmholtz system \eqref{eq:trans1} carry the geometric information of the underlying domain $\Omega$. In fact, it is shown that there exists an infinite sequence of transmission eigenfunctions whose $L^2$-energies concentrate on $\partial\Omega$. This geometrical property has been further used to develop a super-resolution wave imaging scheme associated with the acoustic scattering. In this paper, we shall show that the Maxwell transmission eigenfunctions in \eqref{eq:tcase12} carry both the geometrical and topological properties of the underlying domain $\Omega$. First, we construct sequences of eigenfunctions $(\mathbf{E}_1^{(n)}, \mathbf{E}_2^{(n)})_{n\in\mathbb{N}_+}$ associated with eigenvalues $k_n\rightarrow\infty$ whose $L^2$-energies localize around $\partial\Omega$. In fact, we construct both mono-localized and bi-localized eigenfunctions. Here, by a mono-localized eigenfunction $(\mathbf{E}_1^{(n)}, \mathbf{E}_2^{(n)})$, we mean that only one of the pair of eigenfunctions is surface-localized while the other one is not, and by bi-localization, we mean both $\mathbf{E}_1^{(n)}$ and $\mathbf{E}_2^{(n)}$ are surface-localized. In the radial case with constant refractive index, we give the analytical construction, whereas for the general case, we provide numerical verifications. Second, we show that the transmission eigenfunctions carry the topological structure of the underlying domain $\Omega$ when it has multiple connected components. Using the obtained geometrical and topological results, we propose a novel and interesting application of generating the mirage imaging of a given optical object at a customized location. The main idea is to design suitable illuminating EM fields that are related to the underlying object as well as the associated EM transmission eigenfunctions.

The rest of the paper is organized as follows. Sections 2 and 3 are respectively devoted to the geometrical and topological properties of the EM transmission eigenfunctions. In Section 4, we consider the application to artificial mirage.

\section{Surface-localized transmission eigenmodes}

To begin with, we introduce a qualitative definition of surface-localization for a function $\mathbf{Q}\in L^2(\Omega)^3$. In what follows, for a sufficiently small $\delta\in\mathbb{R}_+$, we set
  \begin{equation*}\label{eq:neigh1}
    \mathcal{N}_{\delta}(\partial \Omega):=\{\mathbf{x}\in \Omega;  \ \mathrm{dist}(\mathbf{x},\partial \Omega)<\delta\}.
  \end{equation*}
\begin{defn}\label{def:1}
  Given a function $ \mathbf{Q}\in L^2(\Omega)^3$. It is said to be surface-localized if there exists a sufficient small constant  $\delta\in\mathbb{R}_+$, such that
  \begin{equation}\label{eq:defn1}
    \frac{\|\mathbf{Q}\|_{L^2(\Omega\backslash\mathcal{N}_{\delta}(\partial \Omega))^3}}{\|\mathbf{Q}\|_{L^2(\Omega)^3}}\ll 1.
  \end{equation}
\end{defn}
By \eqref{eq:defn1}, if a function is surface-localized, then its $L^2$-energy concentrates on $\partial\Omega$. In what follows, we shall make the qualitative relation in \eqref{eq:defn1} more quantitative.

Our main finding in this section can be summarized in the following theorem.

\begin{thm}\label{thm:1}
Let $\Omega$ be a simply connected Lipschitz domain and $\varepsilon$ be a positive constant with $\varepsilon\neq 1$. Consider the Maxwell transmission eigenvalue problem \eqref{eq:tcase1}. Then there exists a sequence of mono-localized eigenfunctions $(\mathbf{E}_1^{(n)}, \mathbf{E}_2^{(n)})_{n\in\mathbb{N}_+}$, and there also exists a sequence of bi-localized eigenfunctions $(\mathbf{E}_1^{(n)}, \mathbf{E}_2^{(n)})_{n\in\mathbb{N}_+}$, associated with the transmission eigenvalues $k_n\rightarrow\infty$.
\end{thm}

In what follows, we shall rigorously prove Theorem~\ref{thm:1} in the case that $\Omega$ is radially symmetric, whereas in the general case, we only provide numerical verifications. In the numerics, we can actually consider variable $\varepsilon$ and the conclusion in Theorem~\ref{thm:1} still holds true.

\subsection{Transmission eigenvalues for a ball}
In this subsection, we briefly discuss the transmission eigenvalues for the Maxwell system \eqref{eq:tcase12} in the radially symmetric case associated with a constant refractive index. This has been derived in \cite{Monk2012}. Nevertheless, some of the technical ingredients shall be needed in our subsequent analysis of the transmission eigenfunctions and for self-containedness and easy reference, we present them in what follows. Without loss of generality, we assume that $\Omega$ centres at the origin, namely, $\Omega=\{\mathbf{x}\in \mathbb{R}^3:\,|\mathbf{x}|<R\in\mathbb{R}_+\}$. We also let $\epsilon_0=\sqrt{\varepsilon}\in\mathbb{R}_+$. Then \eqref{eq:tcase12} can be rewritten as
\begin{equation}\label{eq:tcase12r}
\begin{cases}
\begin{aligned}
&\nabla\wedge(\nabla\wedge\mathbf{E}_1)-k^2\epsilon_0^2\mathbf{E}_1=0, && \nabla\cdot (\epsilon_0^2 \mathbf{E}_1)=0 &&\mbox{in}\ \ \Omega,\medskip\\
&\nabla\wedge(\nabla\wedge\mathbf{E}_2)-k^2\mathbf{E}_2=0, && \nabla\cdot\mathbf{E}_2=0 &&\mbox{in}\ \ \Omega,\medskip\\
&\nu\wedge\mathbf{E}_1=\nu\wedge\mathbf{E}_2, && \nu\wedge(\nabla\wedge \mathbf{E}_1)=\nu\wedge(\nabla\wedge\mathbf{E}_2) && \mbox{on}\ \ \partial\Omega.
\end{aligned}
\end{cases}
\end{equation}
Since $\epsilon_0$ is a positive constant, we can expand the solutions $\mathbf{E}_1$ and $\mathbf{E}_2$ of the system \eqref{eq:tcase12r} into the Fourier series\cite{CK},
\begin{equation*}
\begin{aligned}
  &{\bf E}_1(\mathbf{x})=\sum_{n=1}^{\infty}\sum_{m=-n}^{n} a_n^m \mathcal{M}_n^m(\mathbf{x})
  +\sum_{n=1}^{\infty}\sum_{m=-n}^{n} b_n^m \mathcal{N}_n^m(\mathbf{x}) ,\medskip\\
 &{\bf E}_2(\mathbf{x})=\sum_{n=1}^{\infty}\sum_{m=-n}^{n} \widetilde{a}_n^m \widetilde{\mathcal{M}}_n^m(\mathbf{x})
  +\sum_{n=1}^{\infty}\sum_{m=-n}^{n} \widetilde{b}_n^m \widetilde{\mathcal{N}}_n^m(\mathbf{x}) ,\medskip\\
  \end{aligned}
  \end{equation*}
  where
  \begin{equation}\label{eq:MN}
    \begin{aligned}
    &\mathcal{M}_{n}^m(\mathbf{x})=\nabla\wedge\{\mathbf{x} j_n(k \epsilon_0|\mathbf{x}|)Y_n^m(\hat{\mathbf{x}})\}, \quad \mathcal{N}_{n}^m(\mathbf{x})=\frac{1}{\mathrm{i}k}\nabla\wedge\mathcal{M}_{n}^m(\mathbf{x}),\\
    &\widetilde{\mathcal{M}}_{n}^m(\mathbf{x})=\nabla\wedge\{\mathbf{x} j_n(k|\mathbf{x}|)Y_n^m(\hat{\mathbf{x}})\}, \quad \ \ \, \widetilde{\mathcal{N}}_{n}^m(\mathbf{x})=\frac{1}{\mathrm{i}k}\nabla\wedge\widetilde{\mathcal{M}}_{n}^m(\mathbf{x}).
    \end{aligned}
  \end{equation}
  Here,  $j_n(\mathbf{x})$ denotes the spherical Bessel functions of the first kind and $ Y_n^{m} $ denotes a spherical harmonic function of degree $n$ and order $m$.
  Moreover, one can verify that $\mathcal{M}_{n}^m $ and $\mathcal{N}_{n}^m$ are the solutions to
  \begin{equation*}
  \nabla\wedge(\nabla\wedge\mathbf{E}_1)-k^2\epsilon_0^2\mathbf{E}_1=0, \quad
  \nabla\cdot(\epsilon_0^2 \mathbf{E}_1)=0,
  \end{equation*}
 and $\widetilde{\mathcal{M}}_{n}^m $ and $\widetilde{\mathcal{N}}_{n}^m$ are the solutions to
\begin{equation*}
 \nabla\wedge(\nabla\wedge\mathbf{E}_2)-k^2\mathbf{E}_2=0,\ \ \ \ \nabla\cdot\mathbf{E}_2=0.
\end{equation*}
In particular,
  $\mathcal{M}_{n}^m$ and $\widetilde{\mathcal{M}}_{n}^m$ denote the TE (transverse electric) waves; $\mathcal{N}_{n}^m$ and $\widetilde{\mathcal{N}}_{n}^m$ denote the TM (transverse magnetic) waves.

For the TE modes, using the spherical coordinates $(\rho, \theta, \phi)$, $\mathcal{M}_{n}^m$ can be rewritten as
\begin{equation}\label{eq:M}
\begin{aligned}
  \mathcal{M}_{n}^m&=\nabla\wedge\{\rho j_n(k\epsilon_0\rho)Y_n^m(\theta,\phi)\,\hat{\bm{\rho}}\}\\
  &=\frac{1}{\rho^2 \sin\theta }
        \left|\begin{array}{ccc}
        \hat{\bm{\rho}} &  \rho \hat{\bm{\theta}}    & \rho \sin\theta \hat{\bm{\phi}}\medskip  \\
         \frac{\partial}{\partial \rho} & \frac{\partial}{\partial \theta} & \frac{\partial}{\partial \phi} \medskip\\
         \rho j_n(k\epsilon_0\rho) Y_n^m(\theta,\phi) & 0 & 0
       \end{array}\right| \\
  &=\frac{1}{\sin\theta} j_n(k\epsilon_0\rho) \frac{\partial Y_n^m(\theta, \phi)}{\partial \phi}\hat{\bm{\theta}}-j_n(k\epsilon_0\rho) \frac{\partial Y_n^m(\theta, \phi)}{\partial \theta}\hat{\bm{\phi}},
  \end{aligned}
\end{equation}
where $\rho=|\mathbf{x}|$ and
\begin{equation*}
\begin{aligned}
  &\hat{\bm{\rho}}=\sin\theta\cos \phi\,\bm e_1 + \sin\theta\sin \phi\,\bm e_2+ \cos\theta\, \bm e_3,\\
  &\hat{\bm{\theta}}=\cos\theta\cos \phi\,\bm e_1 + \cos\theta\sin \phi\,\bm e_2-\sin\theta\, \bm e_3,\\
  &\hat{\bm{\phi}}=-\sin \phi\,\bm e_1 + \cos \phi\,\bm e_2.
  \end{aligned}
\end{equation*}
Due to
\begin{equation*}
\frac{1}{\sin\theta} \frac{\partial}{\partial \theta} \left(\sin \theta \frac{\partial Y_n^m(\theta, \phi)}{\partial \theta}\right)+
  \frac{1}{\sin^2 \theta} \frac{\partial^2 Y_n^m(\theta, \phi)}{\partial \phi^2}+n(n+1)Y_n^m(\theta, \phi)=0,
\end{equation*}
 we can derive that
\begin{equation}\label{eq:curl-M}
\begin{aligned}
  &\nabla \wedge \mathcal{M}_{n}^m
  =\frac{1}{\rho^2 \sin\theta }
        \left|\begin{array}{ccc}
        \hat{\bm{\rho}} &  \rho \hat{\bm{\theta}}    & \rho \sin\theta \hat{\bm{\phi}}\medskip  \\
         \frac{\partial}{\partial \rho} & \frac{\partial}{\partial \theta} & \frac{\partial}{\partial \phi} \medskip\\
         0 & \rho\frac{1}{\sin\theta} j_n(k\epsilon_0\rho) \frac{\partial Y_n^m(\theta, \phi)}{\partial \phi} & -\rho \sin \theta j_n(k\epsilon_0\rho) \frac{\partial Y_n^m(\theta, \phi)}{\partial \theta}
       \end{array}\right| \\
  &=-\frac{j_n(k\epsilon_0\rho)}{\rho }\left(\frac{1}{\sin\theta} \frac{\partial}{\partial \theta}\left(\sin \theta \frac{\partial Y_n^m(\theta, \phi)}{\partial \theta}\right)+
  \frac{1}{\sin^2 \theta} \frac{\partial^2 Y_n^m(\theta, \phi)}{\partial \phi^2}\right)\hat{\bm{\rho}}\\
  &\quad \quad +\frac{1}{\rho}\frac{\partial}{\partial \rho}(\rho j_n(k\epsilon_0\rho) )\frac{\partial Y_n^m(\theta, \phi)}{\partial \theta} \hat{\bm{\theta}}
  +\frac{1}{\rho\sin\theta}\frac{\partial}{\partial \rho}(\rho j_n(k\epsilon_0\rho) )\frac{\partial Y_n^m(\theta, \phi)}{\partial \phi} \hat{\bm{\phi}}\\
  &=\frac{1}{\rho} j_n(k\epsilon_0\rho) n (n+1) Y_n^m(\theta, \phi) \hat{\bm{\rho}}+\frac{\partial}{\partial \rho}(\rho j_n(k\epsilon_0\rho) )\frac{1}{\rho}\frac{\partial Y_n^m(\theta, \phi)}{\partial \theta} \hat{\bm{\theta}}\\
 &\quad\quad  +\frac{\partial}{\partial \rho}(\rho j_n(k\epsilon_0\rho) )\frac{1}{\rho\sin\theta}\frac{\partial Y_n^m(\theta, \phi)}{\partial \phi} \hat{\bm{\phi}}.
  \end{aligned}
\end{equation}
Since the surface gradient in  the spherical coordinates is given by
\begin{equation*}
   \nabla_s f=\frac{\partial f}{\partial \theta}\hat{\bm{\theta}} +\frac{1}{\sin\theta} \frac{\partial f}{\partial \phi}\hat{\bm{\phi}},
\end{equation*}
equations \eqref{eq:M} and \eqref{eq:curl-M} can be rewritten as
\begin{equation}\label{eq:M-reduce}
\begin{aligned}
\mathcal{M}_{n}^m & =j_n(k\epsilon_0\rho) \nabla_s Y_n^m(\theta, \phi) \wedge \hat{\bm{\rho}},\\
\nabla \wedge \mathcal{M}_{n}^m&=\frac{1}{\rho} j_n(k\epsilon_0\rho) n (n+1) Y_n^m(\theta, \phi)\hat{\bm{\rho}}+ \frac{1}{\rho}\frac{\partial}{\partial \rho}(\rho j_n(k\epsilon_0\rho)) \nabla_s Y_n^m(\theta, \phi).
  \end{aligned}
\end{equation}
Similarly, one can deduce that
\begin{equation}\label{eq:M-tilde}
\begin{aligned}
\widetilde{\mathcal{M}}_{n}^m & =j_n(k\rho) \nabla_s Y_n^m(\theta, \phi) \wedge \hat{\bm{\rho}},\\
\nabla \wedge \widetilde{\mathcal{M}}_{n}^m&=\frac{1}{\rho} j_n(k\rho) n (n+1) Y_n^m(\theta, \phi)\hat{\bm{\rho}}+ \frac{1}{\rho}\frac{\partial}{\partial \rho}(\rho j_n(k\rho)) \nabla_s Y_n^m(\theta, \phi).
  \end{aligned}
\end{equation}
According to the boundary condition to \eqref{eq:tcase12r},  it yields
\begin{equation*}
  \begin{aligned}
  &a_n^m \hat{\bm{\rho}} \wedge \mathcal{ M}_n^m=\widetilde{a}_n^m \hat{\bm{\rho}} \wedge\widetilde{\mathcal{ M}}_n^m,\\
  &a_n^m \hat{\bm{\rho}} \wedge (\nabla \wedge  \mathcal{ M}_n^m)=\widetilde{a}_n^m  \hat{\bm{\rho}} \wedge  (\nabla  \wedge \widetilde{\mathcal{ M}}_n^m).
  \end{aligned}
\end{equation*}
Applying \eqref{eq:M-reduce} and \eqref{eq:M-tilde} with a straightforward calculation,  one can obtain
\begin{equation*}
  \begin{aligned}
  &\Big (a_n^m j_n(k\epsilon_0 \rho)-\widetilde{a}_n^m  j_n(k \rho)\Big)\,  \nabla_s Y_n^m(\theta,\phi)=0,\\
  &\left (a_n^m\frac{1}{\rho}\frac{ \partial}{\partial \rho} (\rho j_n(k\epsilon_0\rho))-\widetilde{a}_n^m \frac{1}{\rho}\frac{ \partial}{\partial \rho} (\rho j_n(k\rho) )\right) \hat{\rho} \wedge\nabla_s Y_n^m(\theta,\phi)=0, \quad \rho=R.\\
  \end{aligned}
\end{equation*}
{Noting that  $a_n^m\neq 0$ and $\widetilde{a}_n^m\neq 0$, the eigenvalues $k$'s of the TE modes are positive zeros of the following function}
\begin{equation}\label{eq:kFunc}
  f_n^{\mathrm{TE}}(k)=\left(j_n(k\epsilon_0\rho) \frac{1}{\rho}\frac{ \partial}{\partial \rho} (\rho j_n(k\rho) ) -j_n(k\rho) \frac{1}{\rho}\frac{ \partial}{\partial \rho} (\rho j_n(k\epsilon_0\rho) )\right)_{\rho=R}.
\end{equation}
Using  the recurrence relation of the derivatives of the spherical Bessel functions,
\begin{equation}\label{eq:recurrence}
  \begin{aligned}
  &\frac{\partial j_n(k\rho)}{\partial \rho}=k\left(\frac{n}{k\rho} j_n(k\rho)-j_{n+1}(k\rho) \right),\\
  &\frac{\partial j_n(k\epsilon_0\rho)}{\partial \rho}=k\epsilon_0\left(\frac{n}{k\epsilon_0\rho} j_n(k\epsilon_0\rho)-j_{n+1}(k\epsilon_0\rho) \right),\\
  \end{aligned}
\end{equation}
equation \eqref{eq:kFunc} can be rewritten as
\begin{equation}\label{eq:TEf}
    f_n^{\mathrm{TE}}(k)=k \Big (\epsilon_0 j_n(k R)  j_{n+1}(k\epsilon_0 R)-j_{n+1}(k R) j_n(k\epsilon_0 R)  \Big),\quad n\geq 1.
\end{equation}
On the other hand, for the TM modes, one can derive that
\begin{equation*}
\begin{aligned}
  \mathcal{N}_{n}^m(\mathbf{x})&=\frac{1}{\mathrm{i}k}\nabla\wedge\mathcal{M}_{n}^m(\mathbf{x})\\
  &=\frac{1}{\mathrm{i}k\rho} j_n(k\epsilon_0\rho) n (n+1) Y_n^m(\theta, \phi)\hat{\bm{\rho}}+ \frac{1}{\mathrm{i}k\rho}\frac{\partial}{\partial \rho}(\rho j_n(k\epsilon_0\rho)) \nabla_s Y_n^m(\theta, \phi).
  \end{aligned}
\end{equation*}
Noting that
$\nabla \wedge\left(\nabla\wedge\mathcal{M}_{n}^m\right)-k^2 \epsilon_0^2 \mathcal{M}_{n}^m=0$,
together with \eqref{eq:MN} and \eqref{eq:M-reduce}, one can  derive that
\begin{equation*}
 \nabla \wedge \mathcal{N}_{n}^m
= \nabla \wedge\left(\frac{1}{\mathrm{i}k}\nabla\wedge\mathcal{M}_{n}^m\right)
=-\mathrm{i}k\epsilon_0^2\mathcal{M}_{n}^m=-\mathrm{i}k\epsilon_0^2j_n(k\epsilon_0\rho) \nabla_s Y_n^m(\theta, \phi) \wedge \hat{\bm{\rho}}.
\end{equation*}
Similar results also hold for $\widetilde{\mathcal{N}}_{n}^m$.  According to last equation of  \eqref{eq:tcase12r},  it yields
\begin{equation*}
  \begin{aligned}
  &b_n^m \hat{\bm{\rho}} \wedge \mathcal{ N}_n^m=\widetilde{b}_n^m \hat{\bm{\rho}} \wedge\widetilde{\mathcal{ N}}_n^m,\\
  &b_n^m \hat{\bm{\rho}} \wedge (\nabla \wedge  \mathcal{ N}_n^m)=\widetilde{b}_n^m  \hat{\bm{\rho}} \wedge  (\nabla  \wedge \widetilde{\mathcal{ N}}_n^m).
  \end{aligned}
\end{equation*}
By following a similar argument as the TE model case, one has
\begin{equation*}
  \begin{aligned}
  &\frac{1}{\mathrm{i}k}\left (b_n^m\frac{1}{\rho}\frac{ \partial}{\partial \rho} (\rho j_n(k\epsilon_0\rho))-\widetilde{b}_n^m \frac{1}{\rho}\frac{ \partial}{\partial \rho} (\rho j_n(k\rho) )\right) \hat{\bm{\rho}}\wedge\nabla_s Y_n^m(\theta,\phi)=0,\\
  &-\mathrm{i}k \Big (b_n^m \epsilon_0^2 j_n(k \epsilon_0 \rho)-\widetilde{b}_n^m  \,  j_n(k \rho)\Big)  \nabla_s Y_n^m(\theta,\phi)=0,\quad \rho=R.
  \end{aligned}
\end{equation*}
Noting that $b_n^m\neq 0$ and $\widetilde{b}_n^m\neq 0$, the eigenvalues $k$'s for TM modes are positive zeros of the following function
\begin{equation*}
f_n^{\mathrm{TM}}(k)=\left(j_n(k\rho)\frac{1}{\rho} \frac{ \partial}{\partial \rho} (\rho j_n(k\epsilon_0\rho) ) -\epsilon_0^2 j_n(k\epsilon_0\rho) \frac{1}{\rho}\frac{ \partial}{\partial \rho} (\rho j_n(k\rho) )\right)_{\rho=R}.
\end{equation*}
Using the recurrence relation \eqref{eq:recurrence}, the last equation can be rewritten as
\begin{equation}\label{eq:TMf}
\begin{aligned}
    f_n^{\mathrm{TM}}(k)=
    &\frac{(1-\epsilon_0^2)(1+n)}{\rho}j_n(k R)j_n(k\epsilon_0 R)\\
     &+k\epsilon_0 \Big (\epsilon_0 j_n(k\epsilon_0 R) j_{n+1}(k R) - j_n(k R)  j_{n+1}(k\epsilon_0 R) \Big),\quad n\geq 1.
    \end{aligned}
\end{equation}

Based on the above discussion,  the transmission eigenvalues are those $k$'s (nontrivial solutions) for  $f_n^{\mathrm{TE}}(k)=0$ or $f_n^{\mathrm{TM}}(k)=0$ for $n\in \mathbb{N}_+$. In particular,
from \eqref{eq:TEf} and \eqref{eq:TMf},  one can find that transmission eigenvalues depend on the parameter $n$, namely the order of the spherical Bessel functions $j_n$.

\subsection{Mono-localized transmission eigenmodes} \label{sec:Mono-localized}
In this section, we construct a sequence of transmission eigenvalues $\{k_n\}_{n\in \mathbb{N}_+}$ and  prove that one of the corresponding pair of eigenfunctions $\left(\mathbf{E}_1^{(n)}, \mathbf{E}_2^{(n)}\right)$  is surface-localized while the other is not. In what follows, for simplification,  we only consider the case with $\epsilon_0>1$ and the case of $0<\epsilon_0<1$ can be deduced by a similar argument.

\begin{lem}\label{lem:eigenvalues}
Let  $\Omega=\{\mathbf{x}\in \mathbb{R}^3: |\mathbf{x}|<R\in\mathbb{R}_+\}$ and $\epsilon_0>1$ be a constant. Let $E_\Omega$ denote the set of transmission eigenvalues of \eqref{eq:tcase12r}. Then there exists a sequence  $\{k_n\}_{n\in\mathbb{Z}_+}\subset E_\Omega$, such that for sufficiently large $n$ , there holds
\begin{equation}\label{eq:eigspan01}
k_n\in \left(\frac{r_{n,s_1(n)}}{R}, \frac{r_{n,s_2(n)}}{R}\right),
\end{equation}
where $r_{n,s}$ denote  the $s$-th positive root of the spherical Bessel function $j_n(|\mathbf{x}|)$ for a fixed order $n$,
\begin{equation}\label{eq:choices01}
s_1(n):=\left[\left(n+\frac{1}{2}\right)^{\gamma_1}\right], \quad s_2(n)=\left[\left(n+\frac{1}{2}\right)^{\gamma_2}\right], \quad 0<\gamma_1<\gamma_2<1,
\end{equation}
and $[\,t\,]$ signifies the integer part of a real number $t$.  Moreover, one has
\begin{equation}\label{eq:leasymk01}
 k_n>\frac{n+\frac{1}{2}}{R}, \quad n=1,2,3,\cdots.
\end{equation}
\end{lem}
\begin{proof}

For the TE models, from equation \eqref{eq:TEf}, we have
\begin{equation*}\label{eq:func-eigenvalue}
  f_n^{\mathrm{TE}}\left(\frac{r_{n,s_1}}{R}\right) f_n^{\mathrm{TE}}\left(\frac{r_{n,s_2}}{R}\right)
  =\frac{r_{n,s_1}r_{n,s_2}}{R^2}
  j_{n+1}(r_{n,s_1})j_n(r_{n,s_1}\epsilon_0)j_{n+1}(r_{n,s_2})j_n(r_{n,s_2}\epsilon_0).
\end{equation*}
Similarly, for the TM models, from equation \eqref{eq:TMf}, one has
\begin{equation*}
  f_n^{\mathrm{TM}}\left(\frac{r_{n,s_1}}{R}\right) f_n^{\mathrm{TM}}\left(\frac{r_{n,s_2}}{R}\right)
  =\frac{r_{n,s_1}r_{n,s_2}\epsilon_0^4}{R^2}
  j_{n+1}(r_{n,s_1})j_n(r_{n,s_1}\epsilon_0)j_{n+1}(r_{n,s_2})j_n(r_{n,s_2}\epsilon_0).
\end{equation*}
Next, we will show that there exists $s_1$ and $s_2$ such that
\begin{equation*}
  j_{n+1}(r_{n,s_1})j_n(r_{n,s_1}\epsilon_0)j_{n+1}(r_{n,s_2})j_n(r_{n,s_2}\epsilon_0)<0.
\end{equation*}
Based on the relationship between the Bessel and spherical Bessel function
\begin{equation}\label{eq:Spherical-Bessel}
  j_{n}(z)=\sqrt{\frac{\pi}{2z}} J_{n+1/2}(z), \quad z>0,
\end{equation}
one can derive that the  positive root  $r_{n,s}$ of spherical Bessel function $j_n(z)$ has the following sharp upper and lower bounds\cite{Qu2008}
\begin{equation}\label{eq:jms}
  n+\frac{1}{2}-\frac{a_s}{2^{1/3}}\left(n+\frac{1}{2}\right)^{1/3}<r_{n,s}<n+\frac{1}{2}-\frac{a_s}{2^{1/3}}\left(n+\frac{1}{2}\right)^{1/3}+\frac{3}{20}a_s^2\frac{2^{1/3}}{\left(n+\frac{1}{2}\right)^{1/3}},
\end{equation}
where $a_s$ is the $s$-th negative zero of the Airy function and has the representation
\begin{equation*}
  a_s=-\left(\frac{3\pi}{8}(4s-1)  \right)^{2/3}(1+\sigma_s), \quad
  0\leq \sigma_s\leq0.130\left(\frac{3\pi}{8}(4s-1.051)  \right)^{-2}.
\end{equation*}
By noting the choice of $s_1$ and $s_2$ in \eqref{eq:choices01} for sufficiently large $n$, it holds that
\begin{equation}\label{eq:approxj01}
\begin{aligned}
&r_{n,s_1}= \left(n+\frac{1}{2}\right)\left(1+C_0 (n+1/2)^{2(\gamma_1-1)/3}+o((n+1/2)^{2(\gamma_1-1)/3})\right),\\
 &r_{n,s_2}=\left(n+\frac{1}{2}\right)\left(1+C_0 (n+1/2)^{2(\gamma_2-1)/3}+o((n+1/2)^{2(\gamma_2-1)/3})\right),
\end{aligned}
\end{equation}
where $C_0$ is a positive constant.
Note that the Bessel function admits the following asymptotic formula \cite[P. 129]{Kor02}
\begin{equation}\label{eq:approxj02}
J_n(z)=\sqrt{\frac{2}{\pi \sqrt{z^2-n^2}}}\cos\left(\sqrt{z^2-n^2}-\frac{n\pi}{2}+n\arcsin(n/z)-\frac{\pi}{4}\right)\Big(1+o(1)\Big),
\end{equation}
for $z>n$ and $n\rightarrow \infty$.
Combing \eqref{eq:Spherical-Bessel}, \eqref{eq:approxj01},  and \eqref{eq:approxj02}, through a straightforward calculation, one obtains
\begin{equation*}\label{eq:asyjm01}
j_n\left( r_{n,s_i}\epsilon_0 \right)=\frac{1}{C_{n,\epsilon_0}^{(i)}}\cos\left( (n+1/2)\left(\sqrt{\epsilon_0^2-1}-\frac{\pi}{2}+\arcsin\frac{1}{\epsilon_0}
+\mathcal{O}\Big((n+1/2)^{\varsigma_i}\Big)\right) -\frac{\pi}{4} \right)\Big(1+o(1)\Big),
\end{equation*}
where $\varsigma_i:={2(\gamma_i-1)/3},\ i=1,2,$ and
$$
C_{n,\epsilon_0}^{(i)}=\big(n+{1}/{2}\big)\left(\epsilon_0^2(\epsilon_0^2-1)\right)^{-1/4}\left(1+\mathcal{O}\Big(\left(n+{1}/{2}\right)^{\varsigma_i}\Big)\right).
$$
Correspondingly, one can derive that
\begin{equation*}\label{eq:asyjm02}
j_{n+1}\left(r_{n,s_i} \right)=C_{n,1}^{(i)}\cos\left(\,(n+3/2)\mathcal{O}\Big((n+1/2)^{\varsigma_i}\Big) -\frac{\pi}{4} \right)\Big(1+o(1)\Big),
\end{equation*}
where
$$
C_{m,1}^{(i)}=\mathcal{O}\left( \Big( n+\frac{1}{2}\Big)^{1+\epsilon_i/4}\right).
$$
Without loss of generality, we suppose that
$$j_{n-1}\left( r_{n,s_1} \right)j_n\left( r_{n,s_1}\, \epsilon_0 \right)>0.$$
We now show that there exists at least one choice of $s_2=(n+1/2)^{\gamma_2}$ such that
$$j_{n-1}\left( r_{n,s_2} \right)j_n\left( r_{n,s_2}\, \epsilon_0 \right)<0,$$
that is,
\begin{equation}\label{eq:mod1}
\begin{aligned}
&\cos\left((n+3/2)\mathcal{O}\Big((n+1/2)^{\varsigma_2}\Big)-\frac{\pi}{4}\right) \\
&\quad \cdot \cos\left( (n+1/2)\left(\sqrt{\epsilon_0^2-1}-\frac{\pi}{2}+\arcsin\frac{1}{\epsilon_0}
+\mathcal{O}\Big((n+1/2)^{\varsigma_2}\Big)\right) -\frac{\pi}{4} \right)
<0.
\end{aligned}
\end{equation}
Noting that the above two cosine functions never have the same frequency, so it is easy to realize \eqref{eq:mod1} by modifying $\gamma_2$  w.r.t $\varsigma_2(\gamma_2)$. Thus one has
$k_{n}\in(r_{n,s_1(n)}/R, r_{n,s_2(n)}/R)$. \eqref{eq:leasymk01} is then followed trivially from  \eqref{eq:approxj01} and the proof is complete.
\end{proof}

From Lemma \ref{lem:eigenvalues}, we prove that there exists a sequences of transmission eigenvalues. Next, we show the geometrical properties of the the corresponding transmission eigenfunctions $\left(\mathbf{E}_1^{(n)}, \mathbf{E}_2^{(n)}\right)$.

\begin{thm}\label{thm:main}
Let  $\Omega=\{\mathbf{x}\in \mathbb{R}^3: |\mathbf{x}|<R \in\mathbb{R}_+\}$ and $\epsilon_0>1$ be a constant.
Let $(\mathbf{E}_1^{(n)}, \mathbf{E}_2^{(n)})$ be the pair of  transmission eigenfunctions associated with eigenvalue $k_n$ in \eqref{eq:eigspan01} for  Maxwell's transmission eigenvalue problem \eqref{eq:tcase12r}.   Then it holds that the eigenfunction  $\mathbf{E}_2^{(n)}$ is surface-localized but the other eigenfunction $\mathbf{E}_1^{(n)}$ is not surface-localized.
\end{thm}

\begin{proof}
Let $\Omega_{\tau}:=\{\mathbf{x}: |\mathbf{x}|<\tau, \ \tau<R\}$ and $k_n$ be the eigenvalues defined in  \eqref{eq:eigspan01}.
For the TE models, the eigenfunctions are given by
\begin{align}
 \label{eq:E1-TE} &\mathbf{E}_1^{(n)}={a}_n^m \rho\,j_n(k_n\epsilon_0 \rho)\,  \nabla_s Y_n^m(\theta,\phi)\wedge \hat{\bm{\rho}} ,\\
 \label{eq:E2-TE} & \mathbf{E}_2^{(n)}=\widetilde{a}_n^m \rho\,j_n(k_n\rho)\,  \nabla_s Y_n^m(\theta,\phi)\wedge \hat{\bm{\rho}}.
\end{align}
Correspondingly, for the TM models, the eigenfunctions  are  given by
\begin{align}
\label{eq:E1-TM} & \mathbf{E}_1^{(n)}=\frac{{b}_n^m}{\mathrm{i}k_n\rho}\left(j_n(k_n \epsilon_0 \rho) n (n+1) Y_n^m \hat{\bm{\rho}}
  +\frac{\partial}{\partial \rho}\Big(\rho j_n(k_n \epsilon_0 \rho) \Big)\nabla_s Y_n^m(\theta,\phi)\right),\\
 \label{eq:E2-TM} & \mathbf{E}_2^{(n)}
  =\frac{\widetilde{b}_n^m}{\mathrm{i}k_n \rho}\left( j_n(k_n\rho) n (n+1) Y_n^m \hat{\bm{\rho}} +\frac{\partial}{\partial \rho}\Big(\rho j_n(k_n \rho) \Big)\nabla_s Y_n^m(\theta,\phi)\right).
\end{align}

To begin with, for a fixed $\tau$, we  prove that there exists sufficient large $n$ such that $k_n\tau<n+1/2$.
Combining \eqref{eq:eigspan01} and \eqref{eq:approxj01}, we can derive that
\begin{equation*}
k_n= \frac{1}{R}\left(n+\frac{1}{2}\right)\left(1+C_0 (n+1/2)^{\varsigma}\right),\quad \varsigma\in [\varsigma_1,\varsigma_2],
\end{equation*}
where $C_0$ is a positive constant and $\varsigma_i={2(\gamma_i-1)/3}$, $i=1, 2$. Thus, there exists a sufficiently large $n$, such that
\begin{equation}\label{eq:kns}
k_n\tau=\frac{\tau}{R}\left(n+\frac{1}{2}\right)\left(1+C_0 (n+1/2)^{\varsigma}\right)< n+\frac{1}{2},
\end{equation}
for any fixed $\tau<R$.

Next, we prove that  the transmission eigenfunctions $\mathbf{E}_2^{(n)}$ are surface-localized around the boundary $\partial \Omega$.
From \cite[P.129]{Kor02}, one has the following asymptotic formula:
\begin{equation}\label{eq:asympjn01}
J_n(nz)=\frac{z^n \mathrm{e}^{n\sqrt{1-z^2}}}{(2\pi n)^{1/2}(1-z^2)^{1/4}(1+\sqrt{1-z^2})^n}\Big(1+o(1)\Big), \quad 0<z<1.
\end{equation}
By a simple calculation, one can find that
\begin{equation*}
  \frac{z\mathrm{e}^{\sqrt{1-z^2}}}{1+\sqrt{1-z^2}}<1, \quad 0<z<1.
\end{equation*}
Therefore, from \eqref{eq:asympjn01},  we can derive the following asymptotic expansion:
\begin{equation*}
\begin{split}
j_n(k_n\tau)
=&\sqrt{\frac{\pi}{2k_n\tau}} J_{n+1/2}(k_n\tau)\\
=&\frac{1}{2\sqrt{k_n\tau}\left(1-\frac{k_n\tau}{n+1/2}\right)^{1/4}}\left(n+\frac{1}{2}\right)^{-1/2}
\left(\frac{k_n\tau}{n+1/2}\frac{\mathrm{e}^{\sqrt{1-\Big(\frac{k_n\tau}{n+1/2}\Big)^2}}}{1+\sqrt{1-\Big(\frac{k_n\tau}{n+1/2}\Big)^2}}\right)^{n+1/2}\Big(1+o(1)\Big)\\
<&C_1\left(n+\frac{1}{2}\right)^{-1}\Big(1+o(1)\Big),
\end{split}
\end{equation*}
where $C_1$  is a positive constant. According to \cite[p.370]{Abr},  one has
\begin{align}
\label{eq:root} &n+\frac{1}{2}\leq  r_{n,1}^{'}<r_{n,1}<r_{n,2}^{'}<r_{n,2}<r_{n,3}^{'}< \cdots,\\
\label{eq:Bessel-Deri}&j'_{n}(z)=\frac{\sqrt{\pi}(z/2)^{n-1}}{4\Gamma(n+1/2)}\mathop{\Pi} \limits_{s=1}^{\infty}\left(1-\frac{z^2}{r_{n,s}^{'2}}   \right),
\end{align}
where $r_{n,s}^{'}$ denotes the $s$-th positive zero of $j_n^{'}(z)$ that is derivative of $j_n(z)$.
Using  \eqref{eq:kns} and \eqref{eq:root},  one can deduce that
  $k_n\tau <r'_{n,1}$ for sufficiently large  $n$. Thus, by \eqref{eq:Bessel-Deri}, we find that $j_n(k_n\rho)$ is a monotonically increasing with respect to $\rho\in(0,\tau)$. Hence, we obtain
\begin{equation}\label{eq:jn-tau}
 j_n(k_n\rho)<j_n(k_n\tau)<C_1\left(n+\frac{1}{2}\right)^{-1}\Big(1+o(1)\Big),
 \quad 0<\rho<\tau.
\end{equation}
One the other hand, for sufficiently large $n$ , one can choose $\tau_1$
\begin{equation}\label{eq:tau1}
\tau_1=\frac{r_{n,1}'}{k_n},\quad \tau< \tau_1<R.
\end{equation}
Thus, we have
\begin{equation*}
n+\frac{1}{2}<r_{n,1}'=k_n\tau_1.
\end{equation*}
Using the asymptotic expansion \eqref{eq:approxj02}, one can deduce that
\begin{equation}\label{eq:jn-right}
\begin{split}
&\int_\tau^R j_n^2(k_{n}\rho)\, \rho\, \mathrm{d}\rho
\geq\int_{\tau_1}^R j_n^2(k_{n}\rho)\, \rho\, \mathrm{d}\rho\\
\geq&\frac{1}{k_n\sqrt{k_n^2 R^2-(n+\frac{1}{2})^2}}\\
&\ \cdot \int_{\tau_1}^{R}\cos^2\left(\sqrt{k_n^2\rho^2-\left(n+\frac{1}{2}\right)^2}
-\frac{\left(n+\frac{1}{2}\right)\pi}{2}+\left(n+\frac{1}{2}\right)\arcsin\left(\frac{n+\frac{1}{2}}{k_n\rho}\right)-\frac{\pi}{4}\right)\mathrm{d}\rho\\
=&C_2\frac{R(R-\tau_1)}{2}\left(n+\frac{1}{2}\right)^{-1-\varsigma/2}\Big(1+\mathcal{R}(n)\Big),
\end{split}
\end{equation}
where $C_2$ is a positive constant and the remaining term $\mathcal{R}(n)$ has the approximation
\begin{equation*}
\begin{split}
&\lim_{n\rightarrow \infty}\mathcal{R}(n)\\
=&\lim_{n\rightarrow \infty}\int_{\tau_1}^{R}\sin 2\left(\sqrt{k_n^2\rho^2-\left(n+\frac{1}{2}\right)^2}
-\frac{\left(n+\frac{1}{2}\right)\pi}{2}+\left(n+\frac{1}{2}\right)\arcsin\left(\frac{n+\frac{1}{2}}{k_n\rho}\right)\right)\,\mathrm{d}\rho\\
=& 0.
\end{split}
\end{equation*}
Thus, for the TE models, using \eqref{eq:E2-TE}, \eqref{eq:jn-tau} and \eqref{eq:jn-right}, it holds that
  \begin{equation*}
  \begin{aligned}
   \frac{\|\mathbf{E}_2^{(n)}\|_{L^2(\Omega_{\tau} )}^2}{\|\mathbf{E}_2^{(n)}\|^2_{L^2(\Omega)}}
   =& \frac{\int_0^\tau \int_0^{\pi} \int_0^{2 \pi} \Big|j_n(k_n\rho)\, \nabla_s Y_n^m(\theta,\phi)\wedge \bm{\hat{\rho}} \Big|^2\rho^2 \sin\theta\,\mathrm{d}\phi\mathrm{d}\theta\mathrm{d}\rho}
   {\int_0^R \int_0^{\pi} \int_0^{2 \pi} \Big|j_n(k_n\rho)\, \nabla_s Y_n^m(\theta,\phi)\wedge \bm{\hat{\rho}} \Big|^2\rho^2 \sin\theta\,\mathrm{d}\phi\mathrm{d}\theta\mathrm{d}\rho}\\
   \leq&\frac{j_n^2(k_{n}\tau)\int_0^\tau \, \rho^2\, \mathrm{d}\rho}
   {\tau_1 \int_{\tau_1}^R j_n^2(k_{n}\rho)\rho \, \mathrm{d}\rho}\\
   <&\frac{C_1^2\left(n+\frac{1}{2}\right)^{-2}\tau^3/3}
   {C_2\left(n+\frac{1}{2}\right)^{-1-\varsigma/2}\tau_1 R(R-\tau_1)/2}\\
   =&\frac{2C_1^2\tau^3}{3C_2\tau_1 R(R-\tau_1)}\left(n+\frac{1}{2}\right)^{-1+\varsigma/2}\rightarrow 0, \quad \mathrm{as}\quad  n\rightarrow \infty.
   \end{aligned}
  \end{equation*}
  Moreover, for the TM models, according to the recurrence relation of spherical bessel function
  \begin{equation*}
    \frac{\partial }{\partial \rho}j_n(\rho)=\frac{j_{n-1}(\rho)-j_{n+1}(\rho)}{2},
  \end{equation*}
   and formula \eqref{eq:leasymk01}, then the eigenfunction \eqref{eq:E2-TM} can be rewritten as
  \begin{equation}\label{eq:E2-temp2}
  \begin{aligned}
    \mathbf{E}_2^{(n)}
=&\frac{\widetilde{b}_n^m}{\mathrm{i}}\Big( \frac{n (n+1)}{k_n \rho}j_n(k_n\rho) Y_n^m (\theta,\phi) \bm{\hat{\rho}}+ \frac{j_n(k_n\rho)}{k_n \rho} \nabla_s Y_n^m(\theta,\phi) \\ &\quad +\frac{j_{n-1}(k_n\rho)-j_{n+1}(k_n\rho)}{2}\nabla_s Y_n^m(\theta,\phi)\Big)\\
=&\frac{\widetilde{b}_n^m}{\mathrm{i}}\Big( \frac{n (n+1)}{k_n \rho}j_n(k_n\rho) Y_n^m (\theta,\phi) \bm{\hat{\rho}}\\
&\quad+\frac{j_{n-1}(k_n\rho)-j_{n+1}(k_n\rho)}{2}\nabla_s Y_n^m(\theta,\phi)\left(1+\mathcal{O}\left(n^{-1}\right)\right)\Big).
  \end{aligned}
  \end{equation}
%
  Due to the Orthogonality property of the surface gradient to spherical harmonic function
  \begin{equation*}
    \int_{\mathbb{S}^{2}} \nabla_s Y_n^m \cdot \nabla_s Y_{n'}^{m'} \mathrm{d}s =n(n+1)\delta_{nn'}\delta_{mm'},
  \end{equation*}
  together with $\nabla_s Y_n^m(\theta,\phi)=\rho\nabla Y_n^m(\theta,\phi)$ on $\partial \Omega_{\rho}$, and equation \eqref{eq:jn-right}, it holds that
  \begin{equation*}
  \begin{aligned}
\frac{\|\mathbf{E}_2^{(n)}\|_{L^2(\Omega_{\tau} )}^2}{\|\mathbf{E}_2^{(n)}\|^2_{L^2(\Omega)}}
   &=\frac{\int_0^\tau \left(\frac{n(n+1)}{k_n\rho}\,j_n(k_n\rho) \right)^2\rho^2+n(n+1)\left( \frac{j_{n-1}(k_n\rho)-j_{n+1}(k_n\rho)}{2} \right)^2 \rho^2 \, \mathrm{d}\rho}
   {\int_0^R \left(\frac{n(n+1)}{k_n\rho}\,j_n(k_n\rho) \right)^2\rho^2+n(n+1)\left( \frac{j_{n-1}(k_n\rho)-j_{n+1}(k_n\rho)}{2} \right)^2\rho^2 \, \mathrm{d}\rho}\\
   &\leq \frac{\int_0^\tau  \frac{n(n+1)}{k_n^2} j_n^2(k_{n}\rho)+(j_{n-1}^2(k_{n}\rho)+j_{n+1}^2(k_{n}\rho))/2 \,\rho^2 \, \mathrm{d}\rho}
   { \int_{\tau_1}^R \frac{n(n+1)}{k_n^2}j_n^2(k_{n}\rho)\rho \, \mathrm{d}\rho/R}\\
   &<\frac{C_3\left(n+\frac{1}{2}\right)^{-2}\tau}
   {C_2\left(n+\frac{1}{2}\right)^{-1-\varsigma/2}(R-\tau_1)/2}\\
   &=\frac{2C_3\tau}{C_2(R-\tau_1)}\left(n+\frac{1}{2}\right)^{-1+\varsigma/2}\rightarrow 0, \quad \mathrm{as}\quad  n\rightarrow \infty,
   \end{aligned}
  \end{equation*}
where the constant $C_3$ is defined by
\begin{equation*}
  C_3=1+\frac{k_n^2}{n(n+1)}\tau^2.
\end{equation*}
Hence, the transmission eigenfunctions of $\mathbf{E}_2^{(n)}$ for both TE and TM models are surface-localized on the boundary $\partial \Omega$.

Finally,  it remains to prove that  the eigenfunctions $\mathbf{E}_1^{(n)}$ are not
surface-localized around the boundary $\partial \Omega$.
Without loss of generality, we only consider the TE models case and the TM models case can be proved in a similar manner.

From \eqref{eq:tau1}, one can deduce that
\begin{equation*}
\begin{aligned}
  \int_0^{R} j_n^2(k_n\epsilon_0 \rho) \rho^2 \, \mathrm{d}\rho
  &=\frac{1}{\epsilon_0^3}\int_0^{\epsilon_0 R} j_n^2(k_n \rho') \rho'^2 \, \mathrm{d}\rho'\\
  &=\frac{1}{\epsilon_0^3}\int_0^{\frac{r'_{n,1}}{k_n}} j_n^2(k_n \rho') \rho'^2 \, \mathrm{d}\rho'
  +\frac{1}{\epsilon_0^3}\int_{\frac{r'_{n,1}}{k_n}}^{\epsilon_0 R} j_n^2(k_n \rho') \rho'^2 \, \mathrm{d}\rho'\\
  &<\frac{2}{\epsilon_0^3}\int_{\frac{r'_{n,1}}{k_n}}^{\epsilon_0 R} j_n^2(k_n \rho') \rho'^2 \, \mathrm{d}\rho'\\
  &=2\int_{\frac{r'_{n,1}}{\epsilon_0 k_n}}^{ R} j_n^2(k_n\epsilon_0 \rho) \rho^2 \, \mathrm{d}\rho.
  \end{aligned}
\end{equation*}
By following similar arguments, for any $\tau\in (\frac{R}{\epsilon_0}, R)$, then $$\frac{r'_{n,1}}{k_n \epsilon_0}< \frac{R}{\epsilon_0}<\tau <R,$$ and so for $\tau\leq t<R$, one has
\begin{equation*}
\begin{split}
&\int_{\frac{r'_{n,1}}{k_n\epsilon_0}}^{ t} j_n^2(k_n\epsilon_0 \rho) \rho^2 \, \mathrm{d}\rho\\
&=\int_{\frac{r'_{n,1}}{k_n\epsilon_0}}^{ t}\frac{\rho\cos^2\left(\sqrt{k_n^2\epsilon_0^2\rho^2-\left(n+\frac{1}{2}\right)^2}
-\frac{\left(n+\frac{1}{2}\right)\pi}{2}+\left(n+\frac{1}{2}\right)\arcsin\left(\frac{n+\frac{1}{2}}{k_n\rho}\right)-\frac{\pi}{4}\right)}
{k_n\epsilon_0\sqrt{k_n^2 \epsilon_0^2\rho^2-(n+\frac{1}{2})^2}}\mathrm{d}\rho\\
&\approx\frac{1}{k_n^2\epsilon_0^2} \int_{\frac{r'_{n,1}}{ k_n\epsilon_0 }}^{ t} \frac{\rho}{\sqrt{\rho^2-\left(\frac{n+1/2}{k_n\epsilon_0}\right)^2}}\, \mathrm{d}\rho\\
&=\frac{1}{k_n^3\epsilon_0^3} \left( \sqrt{k_n^2\,\epsilon_0^2 \,t^2-\left(n+\frac{1}{2}\right)^2}
-\sqrt{r^{'2}_{n,1}-\left(n+\frac{1}{2} \right)^2}\right).
\end{split}
\end{equation*}
Combining the last two estimations, we obtain
\begin{equation*}
  \begin{aligned}
   \frac{\|\mathbf{E}_1^{(n)}\|_{L^2(\Omega_{\tau} )}^2}{\|\mathbf{E}_1^{(n)}\|^2_{L^2(\Omega)}}
   =& \frac{\int_0^\tau \int_0^{\pi} \int_0^{2 \pi} \Big|j_n(k_n\epsilon_0\rho)\,  \nabla_s Y_n^m(\theta,\phi)\wedge \hat{\bm{\rho}}\Big|^2\rho^2 \sin\theta\,\mathrm{d}\phi\mathrm{d}\theta\mathrm{d}\rho}
   {\int_0^R \int_0^{\pi} \int_0^{2 \pi} \Big| j_n(k_n\epsilon_0\rho)\, \nabla_s Y_n^m(\theta,\phi)\wedge \hat{\bm{\rho}}\Big|^2\rho^2 \sin\theta\,\mathrm{d}\phi\mathrm{d}\theta\mathrm{d}\rho}\\
   \geq&\frac{\int_{\frac{r'_{n,1}}{k_n\epsilon_0}}^{\tau} j_n^2(k_n\epsilon_0 \rho) \rho^2 \, \mathrm{d}\rho}
   {2\int_{\frac{r'_{n,1}}{k_n\epsilon_0}}^{ R} j_n^2(k_n\epsilon_0 \rho) \rho^2 \, \mathrm{d}\rho}\\
   =& \frac{1}{2}\frac{  \sqrt{k_n^2\,\epsilon_0^2 \,\tau^2-\left(n+\frac{1}{2}\right)^2}
-\sqrt{r^{'2}_{n,1}-\left(n+\frac{1}{2} \right)^2}}
{  \sqrt{k_n^2\,\epsilon_0^2 \,R^2-\left(n+\frac{1}{2}\right)^2}
-\sqrt{r^{'2}_{n,1}-\left(n+\frac{1}{2} \right)^2}}\\
\approx &  \frac{1}{2}\frac{\sqrt{\epsilon_0^2(\tau/R)^2-1}}
   {\sqrt{\epsilon_0^2-1}}>0 \quad  n\rightarrow \infty.
   \end{aligned}
  \end{equation*}
The proof is complete.
\end{proof}

\subsection{Bi-localized transmission eigenmodes}\label{sec:Bi-localized}
In this section, we also construct a sequence of transmission eigenvalues $\{k_n\}_{n\in \mathbb{N}_+}$ and  prove that both corresponding pair of eigenfunctions $\left(\mathbf{E}_1^{(n)}, \mathbf{E}_2^{(n)}\right)$ are surface-localized.

\begin{lem}\label{lem:eigenvalues2}
Let  $\Omega=\{\mathbf{x}\in \mathbb{R}^3: |\mathbf{x}|<R\in\mathbb{R}_+\}$ and $\epsilon_0>1$ be a constant. Let $E_\Omega'$ denote the set of transmission eigenvalues of \eqref{eq:tcase12r}. For a fixed value $s_0\in \mathbb{N}_+$,  there exists a subsequence $\{k_n\}_{n\in\mathbb{N}_+}\subset E_\Omega'$,  such that for sufficient large $n$, there holds
\begin{equation}\label{eq:k2}
k_n\in \left(\frac{r_{n,s_0}}{\epsilon_0 R}, \frac{r_{n,s_0+1}}{\epsilon_0 R}\right).
\end{equation}
Moreover, one has
\begin{equation}\label{eq:k2-estimate}
\frac{n+\frac{1}{2}}{\epsilon_0 R}<k_n \leq \frac{n+\frac{1}{2}}{R}, \quad n=1,2,3,\cdots.
\end{equation}
\end{lem}

\begin{proof}
For any fixed $s_0\in \mathbb{N}_+$,  from \eqref{eq:jms} and \eqref{eq:root},  there exists a sufficiently large $n$ such that
\begin{equation}\label{eq:upbound}
  \frac{r_{n,s_0+1}}{\epsilon_0}\leq n+\frac{1}{2}\leq r_{n,1}'.
\end{equation}
By the monotonicity of the Bessel function before the first local maximal, we have
\begin{equation*}
  j_{n}(k)\geq 0, \quad k\in(0, r_{n,1}'].
\end{equation*}
Noting that the positive root of $j_{n}(|x|)$ are interlaced with those of $j_{n+1}(|x|)$ \cite{LZ07}, hence we have
\begin{equation*}
  j_{n+1}(r_{n, s_0})\cdot  j_{n+1}(r_{n, s_0+1})< 0.
\end{equation*}
Therefore, from the last two equations and \eqref{eq:TEf}, we can derive that

\begin{equation*}
\begin{aligned}
   &f_n^{\mathrm{TE}}\left(\frac{r_{n,s_0}}{\epsilon_0 R}\right) f_n^{\mathrm{TE}}\left(\frac{r_{n,s_0+1}}{\epsilon_0 R}\right)\\
  &=\frac{r_{n,s_0}r_{n,s_0+1}}{ R^2}
  j_{n}\left(\frac{r_{n,s_0}}{\epsilon_0}\right)j_{n+1}(r_{n,s_0})\,
   j_{n}\left(\frac{r_{n,s_0+1}}{\epsilon_0}\right)j_{n+1}(r_{n,s_0+1})\\
   &\leq \frac{r_{n,s_0}r_{n,s_0+1}}{R^2}
   j_{n}^2\left(r_{n,1}'\right)j_{n+1}(r_{n,s_0})j_{n+1}(r_{n,s_0+1})<0.
  \end{aligned}
\end{equation*}
Hence, there exists $k_n\in\left({r_{n,s_0}}/{(\epsilon_0 R)}, {r_{n,s_0+1}}/{(\epsilon_0 R)}\right)$
such that $f_n^{\mathrm{TE}}\left(k_n\right)=0$.
By a similar argument, one can also verify that $f_n^{\mathrm{TM}}\left(k_n\right)=0$
and it proves \eqref{eq:k2}.
Furthermore, from \eqref{eq:root} and \eqref{eq:upbound}, one can deduce \eqref{eq:k2-estimate} and it completes the proof.
\end{proof}

\begin{thm}\label{thm:main}
Let  $\Omega=\{\mathbf{x}\in \mathbb{R}^3: |\mathbf{x}|<R \in\mathbb{R}_+\}$ and $\epsilon_0>1$ be a constant. Let $(\mathbf{E}_1^{(n)}, \mathbf{E}_2^{(n)})$ be the pair of  transmission eigenfunctions associated with eigenvalue $k_n$ in \eqref{eq:eigspan01} for  Maxwell's transmission eigenvalue problem \eqref{eq:tcase12r}.   Then it holds that both eigenfunctions $\mathbf{E}_1^{(n)}$ and $\mathbf{E}_2^{(n)}$ are surface-localized around the boundary $\partial \Omega$.
\end{thm}

\begin{proof}
Without loss of generality, we only consider the TE modes and the TM modes can be proved in a similar manner.
Let $\Omega_{\tau}:=\{\mathbf{x}: |\mathbf{x}|<\tau, \ \tau<R\}$ and $k_n$ be the eigenvalues defined in  \eqref{eq:eigspan01}.
 Using \eqref{eq:E2-TE} and \eqref{eq:Spherical-Bessel}, one has
\begin{equation*}
  \begin{aligned}
   \frac{\|\mathbf{E}_2^{(n)}\|_{L^2(\Omega_{\tau} )}^2}{\|\mathbf{E}_2^{(n)}\|^2_{L^2(\Omega)}}
   &= \frac{\int_0^\tau \int_0^{\pi} \int_0^{2 \pi} \Big|j_n(k_n\rho)\, \nabla_s Y_n^m(\theta,\phi)\wedge \hat{\bm{\rho}} \Big|^2\rho^2 \sin\theta\,\mathrm{d}\phi\mathrm{d}\theta\mathrm{d}\rho }
   {\int_0^R \int_0^{\pi} \int_0^{2 \pi} \Big|j_n(k_n\rho)\, \nabla_s Y_n^m(\theta,\phi)\wedge \hat{\bm{\rho}} \Big|^2\rho^2 \sin\theta\, \mathrm{d}\phi\mathrm{d}\theta\mathrm{d}\rho}\\
   &=\frac{\int_{0}^{\tau}  J_{n+\frac{1}{2}}^2(k_{n}\rho)\rho \, \mathrm{d}\rho}
   {\int_{0}^R  J_{n+\frac{1}{2}}^2(k_{n}\rho)\rho \, \mathrm{d}\rho}.
   \end{aligned}
  \end{equation*}
Set $f(\rho)=J_{n+\frac{1}{2}}^2(k_{n}\rho)\rho$, through a straightforward calculation, one can verify that $f(\rho)$ is a convex function on $[0, R]$ for $k_n\leq (n+1/2)/R$.
Therefore, the area of the triangle under tangent of $f(R)$ is smaller than $\int_0^R f(\rho) \mathrm{d}\rho$, that is,
\begin{equation*}
  \frac{R^2 J_{n+\frac{1}{2}}^3(k_{n}R) }{2 J_{n+\frac{1}{2}}(k_{n}R)+4 k_n R J_{n+\frac{1}{2}}^{'}(k_{n}R)}\leq \int_{0}^R  J_{n+\frac{1}{2}}^2(k_{n}\rho)\rho \, \mathrm{d}\rho.
\end{equation*}
Thus it holds that
\begin{equation*}
  \begin{aligned}
   \frac{\|\mathbf{E}_2^{(n)}\|_{L^2(\Omega_{\tau} )}^2}{\|\mathbf{E}_2^{(n)}\|^2_{L^2(\Omega)}}
   &\leq \frac{\tau^2  J_{n+\frac{1}{2}}^2(k_{n}\tau)\Big(2 J_{n+\frac{1}{2}}(k_{n}R)+4 k_n R J_{n+\frac{1}{2}}^{'}(k_{n}R)\Big)}
   {R^2 J_{n+\frac{1}{2}}^3(k_{n}R) }\\
   &\leq \tau^2 \left(\frac{ J_{n+\frac{1}{2}}^2(k_{n}\tau)}
   {J_{n+\frac{1}{2}}^2(k_{n}R)}\right)^2 \left( \frac{2}{R^2}+\frac{4k_n}{R}\cdot \frac{ J_{n+\frac{1}{2}}^{'}(k_{n}R)}
   {J_{n+\frac{1}{2}}(k_{n}R)}  \right).
   \end{aligned}
  \end{equation*}
By Lemma 2.3 in \cite{Deng2021}, one can show that there exists positive constants $C_3$ and $\gamma$ such that
\begin{equation*}
\frac{ J_{n+\frac{1}{2}}^{'}(k_{n}R)}{J_{n+\frac{1}{2}}(k_{n}R)}
  <C_3 \left(n+\frac{1}{2}\right)^{\gamma}.
\end{equation*}
In addition, using the Calini formula, one can  find that  there exists positive constants $C_4$ and $\delta(\tau, \epsilon_0)\in (0,1)$,  such that
\begin{equation*}
  \frac{ J_{n+\frac{1}{2}}^2(k_{n}\tau)}{J_{n+\frac{1}{2}}^2(k_{n}R)}<C_4 (1-\delta)^{n+\frac{1}{2}}.
\end{equation*}
Hence, combining the last three inequalities, one can obtain that
\begin{equation*}
  \frac{\|\mathbf{E}_2^{(n)}\|_{L^2(\Omega_{\tau} )}^2}{\|\mathbf{E}_2^{(n)}\|^2_{L^2(\Omega)}}\rightarrow 0, \quad \mathrm{as} \  n\rightarrow \infty.
\end{equation*}
The proof is complete.
\end{proof}

\begin{rem}
We would like to point out that there is no contradiction between the mono-localized eigenmodes and bi-localized modes, since the transmission eigenvalues for mono-localized eigenmodes satisfy $k_n>\frac{n+1/2}{R}$ while the transmission eigenvalues for bi-localized eigenmodes satisfy
$\frac{n+1/2}{\epsilon_0 R}<k_n<\frac{n+1/2}{R}$. In fact, for a fixed value $n\in \mathbb{N}_+$ and $\epsilon_0>1$,  the mono-localized eigenmodes  would have a better chance for occurring.
\end{rem}

\subsection{Numerics}
In this section,  we present several numerical examples to show geometrical properties of transmission eigenfunctions.
Here, we shall use the curl conforming finite element method approximation as proposed in  \cite{Monk2012} and briefly describe it in the sequel.   Multiplying the first two equation in \eqref{eq:tcase12}  by  a test functions $  \varphi\in H_0^1(\mathrm{curl}, \Omega)$ and integrating by parts,  one obtains
 \begin{equation}\label{eq:weakformula}
 \begin{aligned}
 & (\nabla \wedge \mathbf{E}_1,    \nabla \wedge  \varphi)- k^2 \varepsilon (\mathbf{E}_1,  \varphi)=0, \quad   & \forall  \varphi\in H_0^1(\mathrm{curl}, \Omega),\\
 & (\nabla \wedge \mathbf{E}_2,    \nabla \wedge  \varphi)- k^2(\mathbf{E}_2,  \varphi)=0, \quad  & \forall  \varphi\in H_0^1(\mathrm{curl}, \Omega).
   \end{aligned}
 \end{equation}
To enforce the boundary conditions $\nu \wedge(\nabla\wedge \mathbf{E}_1)=\nu \wedge(\nabla\wedge \mathbf{E}_2)$  weakly, we multiply it by a test function $\psi \in H^1(\mathrm{curl}, \Omega)$ and integrate by parts
\begin{equation}\label{eq:NeumannBoundary}
  (\nabla \wedge \mathbf{E}_1,    \nabla \wedge  \psi)- k^2 \varepsilon (\mathbf{E}_1,  \psi)
  = (\nabla \wedge \mathbf{E}_2,    \nabla \wedge  \psi)- k^2(\mathbf{E}_2,  \psi).
\end{equation}
 Hence,  the variational formulation for  \eqref{eq:tcase12r} is to
  find $(\mathbf{E}_1,  \mathbf{E}_2)\in H(\mathrm{curl},\Omega)\times H(\mathrm{curl},\Omega)$ satisfying  \eqref{eq:weakformula} and  \eqref{eq:NeumannBoundary},
 together with the essential boundary condition $\nu\wedge\mathbf{ E}_1 =\nu\wedge\mathbf{ E}_2$ on $\partial \Omega$.

Let $\mathcal{T}_h$  be a regular  tetrahedral mesh for  $\Omega$ and  $V_h$ be   the curl conforming element space of N\'ed\'elec
\begin{equation*}
  V_h=\{\mathbf{E}\in H(\mathrm{curl},\Omega): \mathbf{E}|_{K}(\bm x)=\bm \alpha + \bm \beta \wedge \bm x,\ K\in \mathcal{T}_h, \ \bm \alpha, \bm \beta\in (P_n)^3  \},
\end{equation*}
where $P_n$ denotes the space of homogeneous polynomials of order $n$.
 We also define a subspace $V_h^0=V\cap H_0(\mathrm{curl},\Omega)$.
 Let
  $\{\bm \eta_i: \, i=1,2,\cdots m  \}$ denote  a basis for $V_h^0$  and   $\{\bm \eta_i: \, i=1,2,\cdots m  \}\cup \{\bm \zeta_j: \, j=1,2,\cdots n \}$ denote   a basis for $V_h$, respectively.   To enforce the boundary condition $\nu \wedge \mathbf{E}_1=\nu \wedge \mathbf{E}_2$,  we set
  \begin{equation*}
   \mathbf{ E}_1^h=\sum_{i=1}^m u_1^{(i)}\bm \eta_i+ \sum_{j=1}^n  w^{(j)}\bm \zeta_j, \quad
    \mathbf{ E}_2^h=\sum_{i=1}^m u_2^{(i)}\bm \eta_i+ \sum_{j=1}^n  w^{(j)}\bm \zeta_j,
  \end{equation*}
 to be the finite element approximations of $ \mathbf{ E}_1$ and $ \mathbf{ E}_2$, respectively.  Thus,  the discrete form of
\eqref{eq:weakformula} and \eqref{eq:NeumannBoundary} can be written as
\begin{equation}\label{eq:Eigenvalue}
  AX=k^2 BX,
\end{equation}
where
\begin{equation*}
  X=[u_1^{(1)}, \cdots, u_1^{(m)}, u_2^{(1)}, \cdots, u_2^{(m)}, w^{(1)}, \cdots, w^{(n)}]\in \mathbb{C}^{2m+n},
\end{equation*}
and
\begin{equation*}
  A=\left[
      \begin{array}{ccc}
        S_{II} & 0 & S_{IB}  \\
        0& S_{II}  &S_{IB} \\
        S_{IB}^{\top} & -S_{IB}^{\top}  & 0 \\
      \end{array}
    \right],
    \quad
    B=\left[
      \begin{array}{ccc}
        M_{2,II} & 0 & M_{2,IB}  \\
        0& M_{1,II}  &M_{1,IB} \\
        M_{2,IB}^{\top} & -M_{1,IB}^{\top} & M_{2,BB}-M_{1,BB} \\
      \end{array}
    \right].
\end{equation*}
Here,   the  stiffness and mass matrices are given by
\begin{equation*}
\begin{aligned}
  &  S_{II}=(\nabla\wedge \bm \eta_i, \nabla\wedge\bm \eta_j),  &&  S_{IB}=(\nabla\wedge\bm \eta_i, \nabla\wedge \bm \zeta_j),
   &&  M_{2,II}=(\varepsilon \bm \eta_i, \bm \eta_j),    &&  M_{1,II}=(\bm \eta_i,\bm  \eta_j),\\
   &   M_{2,IB}=(\varepsilon \bm \eta_i, \bm \zeta_j),    &&  M_{1,IB}=(\bm \eta_i, \bm \zeta_j),  &&  M_{2,BB}=(\varepsilon \bm \zeta_i, \bm \zeta_j),   &&  M_{1,BB}=(\bm \zeta_i, \bm \zeta_j).  \\
   \end{aligned}
\end{equation*}
Here, the assembly of the above matrices is implemented by an open-source PDE solver
FreeFEM \cite{Hecht2012} with the first order curl conforming method. The eigenvalues and eigenvectors for the  non-symmetric eigenvalue problem \eqref{eq:Eigenvalue} are computed  by {\it sptarn} function in MATLAB, which is based on the Arnoldi algorithm with spectral transformation.

\begin{figure}
    \subfigure[$E_1^{(1)}$]{\includegraphics[width=0.32\textwidth]{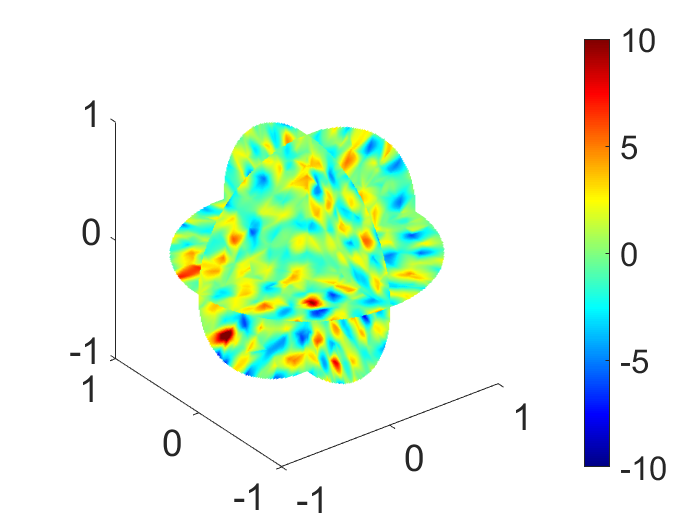}}
    \subfigure[$E_1^{(2)}$]{\includegraphics[width=0.32\textwidth]{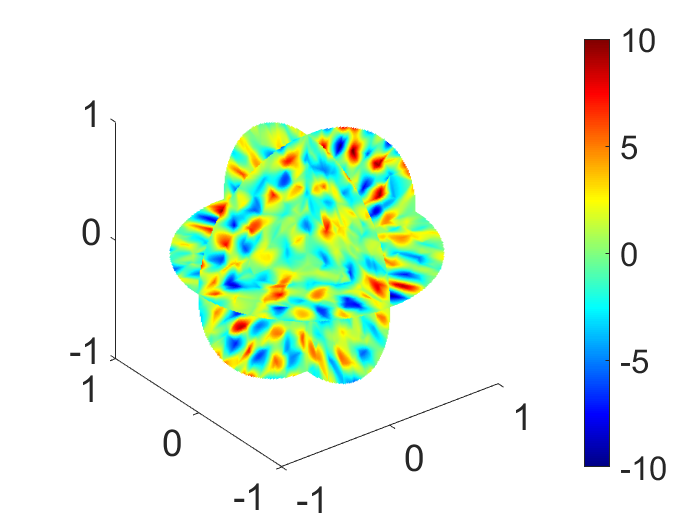}}
    \subfigure[$E_1^{(3)}$]{\includegraphics[width=0.32\textwidth]{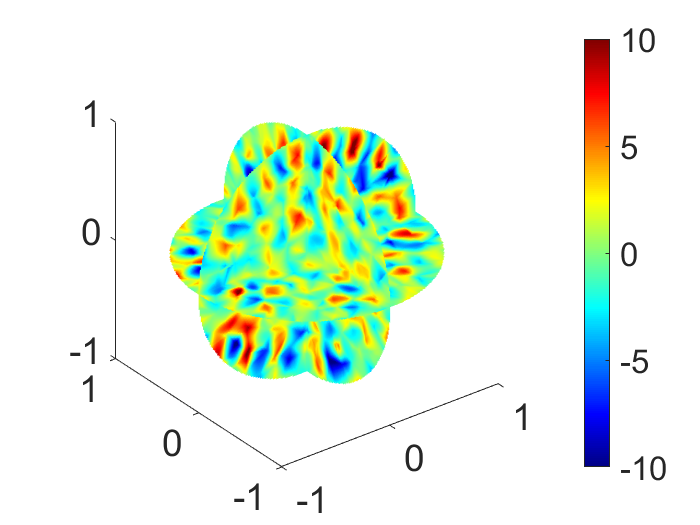}}\\
    \subfigure[$E_2^{(1)}$]{\includegraphics[width=0.32\textwidth]{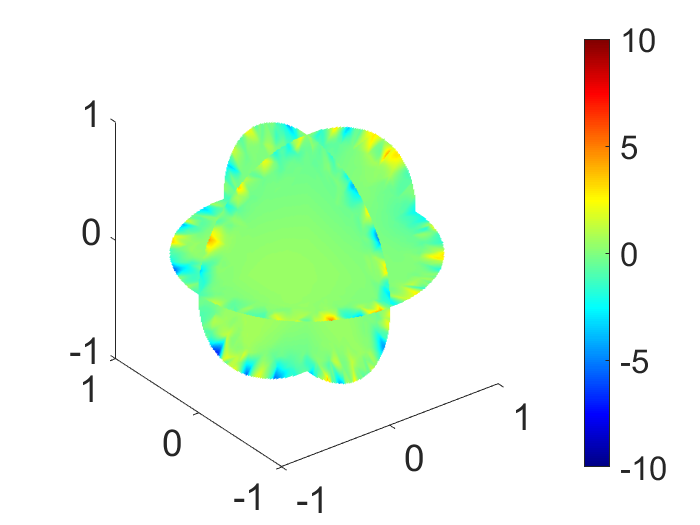}}
    \subfigure[$E_2^{(2)}$]{\includegraphics[width=0.32\textwidth]{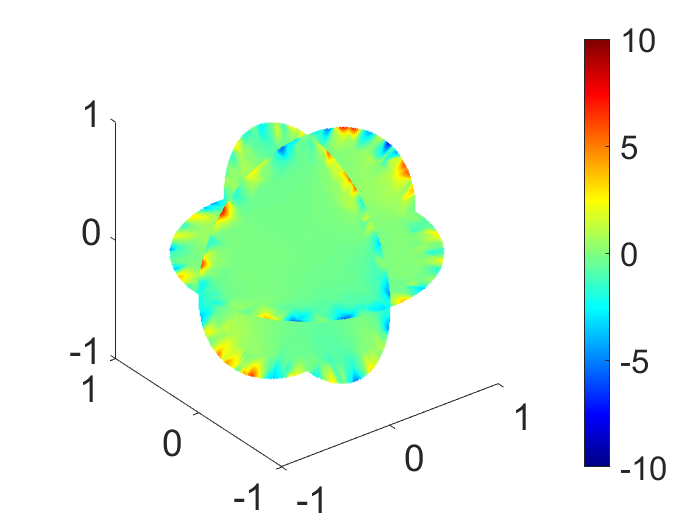}}
    \subfigure[$E_2^{(3)}$]{\includegraphics[width=0.32\textwidth]{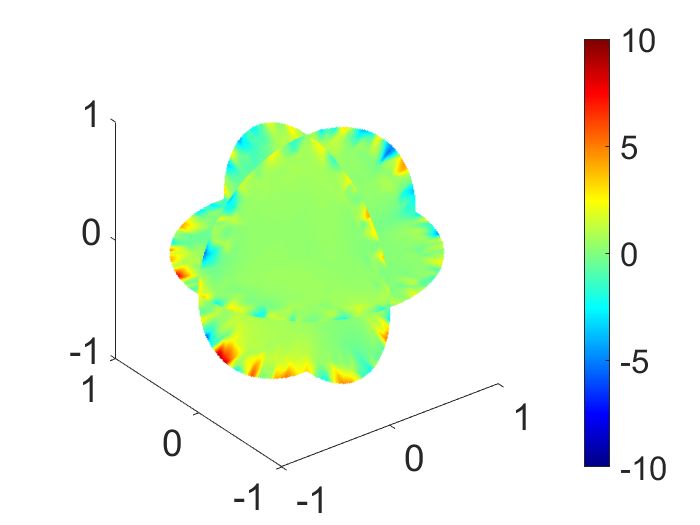}}
    \caption{\label{fig:ball-mono} Mono-localized transmission eigenmodes. Slice plots of eigenfunctions $\mathbf{E}_1$ and $\mathbf{E}_2$ for a unit ball,  where $k=3.3318$. The top row and bottom row  denote the three components of eigenfunctions of $\mathbf{E}_1$ and $\mathbf{E}_2$  respectively.}
\end{figure}

\begin{figure}
    \subfigure[$E_1^{(1)}$]{\includegraphics[width=0.32\textwidth]{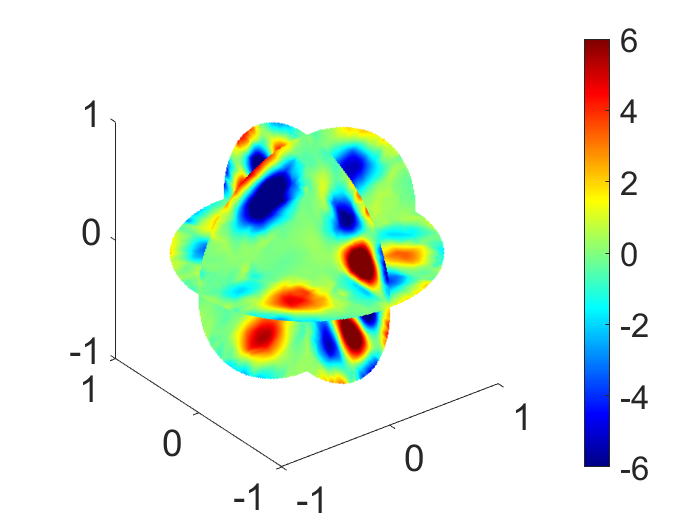}}
    \subfigure[$E_1^{(2)}$]{\includegraphics[width=0.32\textwidth]{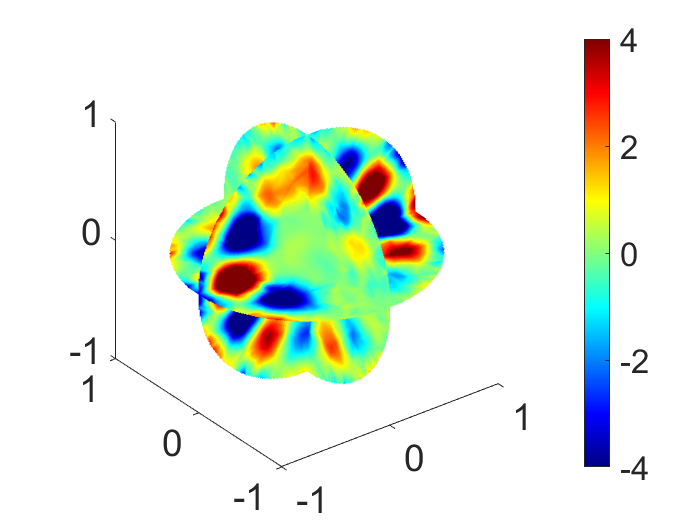}}
    \subfigure[$E_1^{(3)}$]{\includegraphics[width=0.32\textwidth]{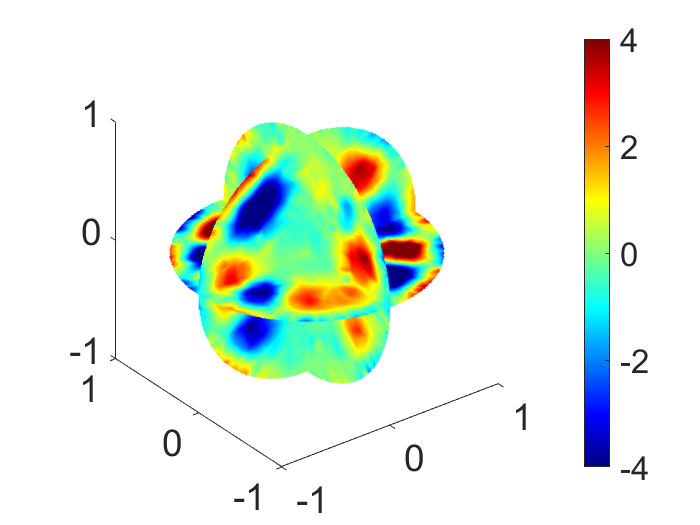}}\\
    \subfigure[$E_2^{(1)}$]{\includegraphics[width=0.32\textwidth]{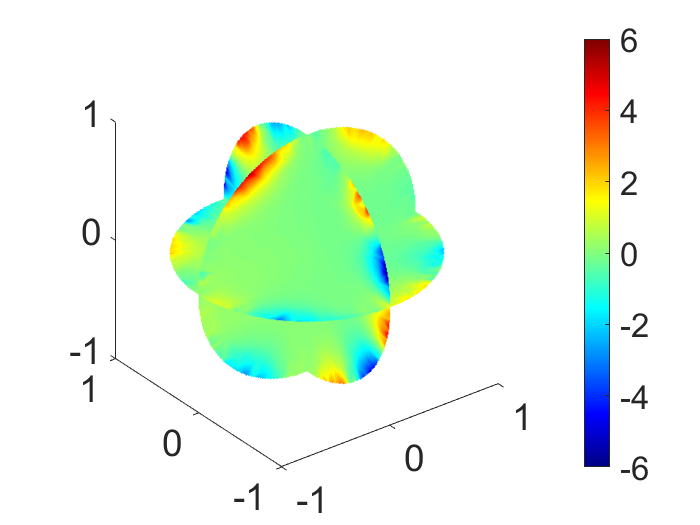}}
    \subfigure[$E_2^{(2)}$]{\includegraphics[width=0.32\textwidth]{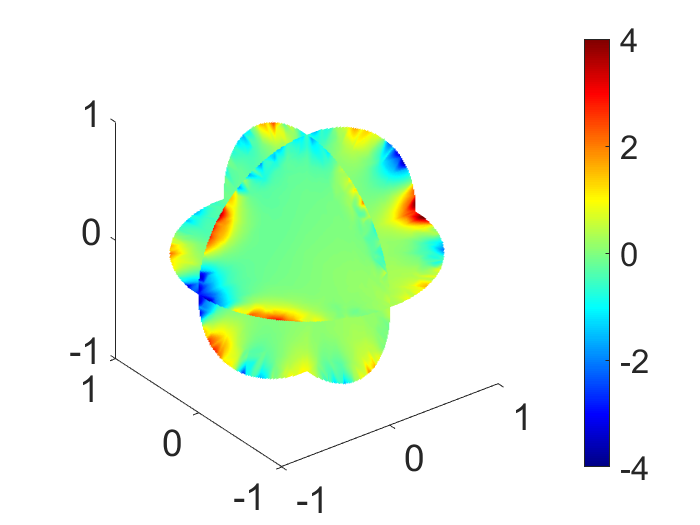}}
    \subfigure[$E_2^{(3)}$]{\includegraphics[width=0.32\textwidth]{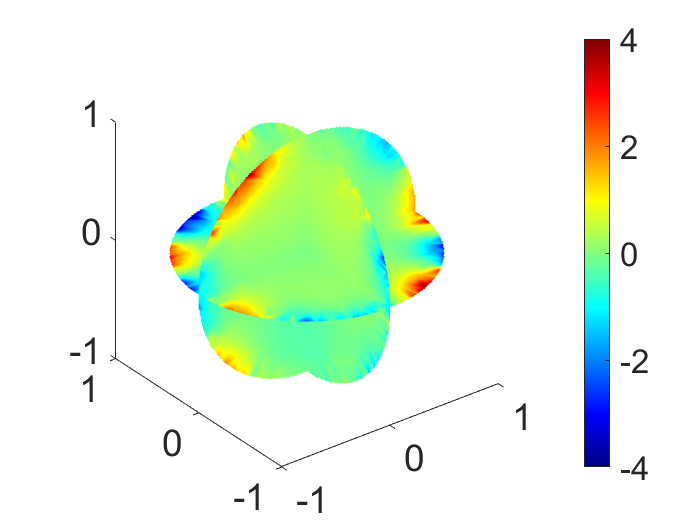}}
    \caption{\label{fig:ball-bi} Bi-localized transmission eigenmodes. Slice plots of eigenfunctions $\mathbf{E}_1$ and $\mathbf{E}_2$ for a unit ball, where $k=3.2296$. The top row and bottom row  denote the three components of eigenfunctions of $\mathbf{E}_1$ and $\mathbf{E}_2$ respectively.}
\end{figure}

\begin{figure}
    \subfigure[$E_1^{(1)}$]{\includegraphics[width=0.32\textwidth]{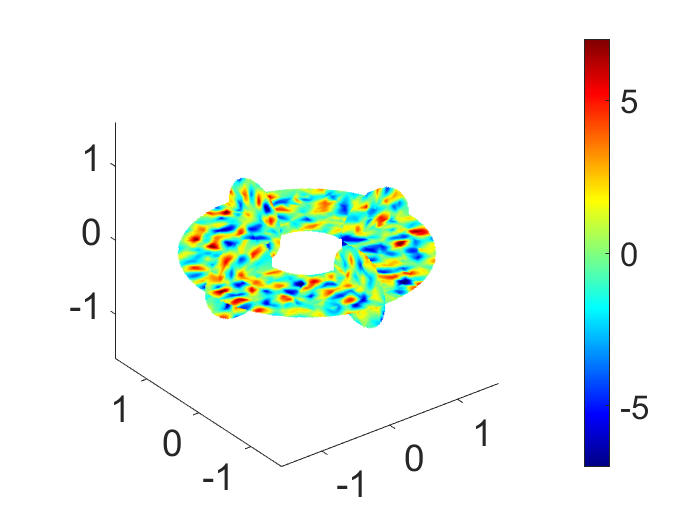}}
    \subfigure[$E_1^{(2)}$]{\includegraphics[width=0.32\textwidth]{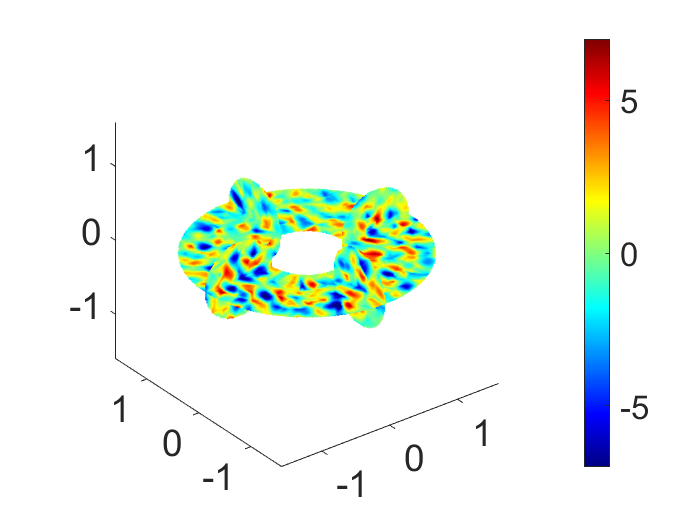}}
    \subfigure[$E_1^{(3)}$]{\includegraphics[width=0.32\textwidth]{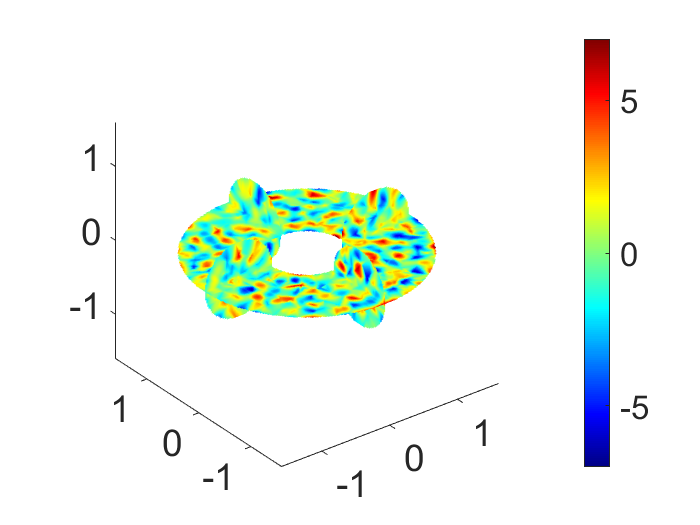}}\\
    \subfigure[$E_2^{(1)}$]{\includegraphics[width=0.32\textwidth]{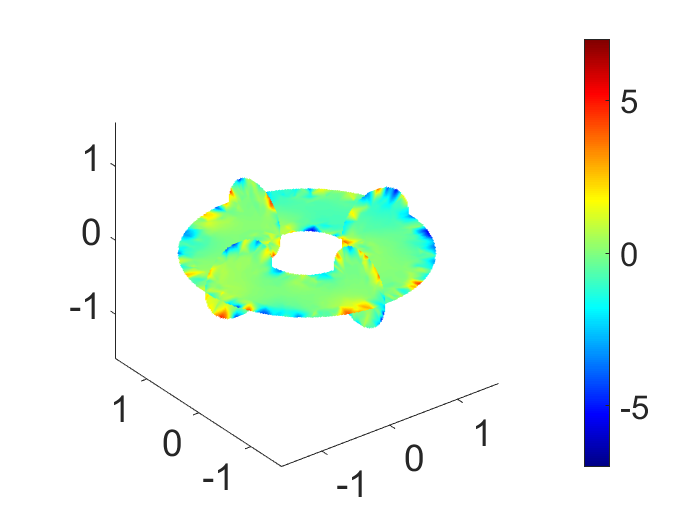}}
    \subfigure[$E_2^{(2)}$]{\includegraphics[width=0.32\textwidth]{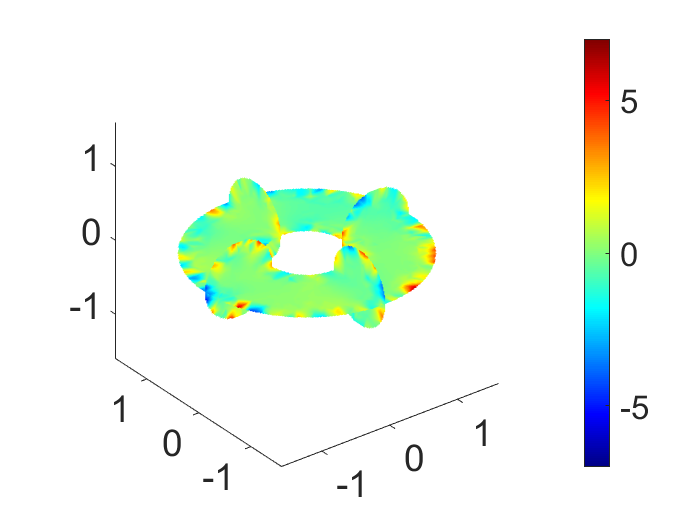}}
    \subfigure[$E_2^{(3)}$]{\includegraphics[width=0.32\textwidth]{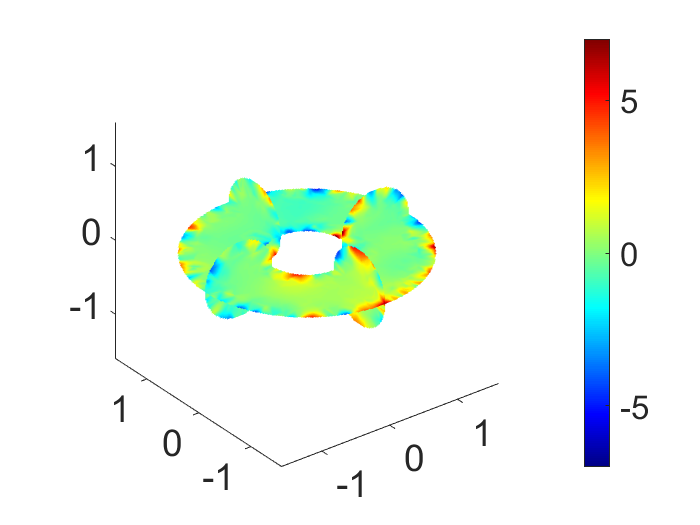}}
    \caption{\label{fig:torus-mono} Mono-localized transmission eigenmodes. Slice plots of eigenfunctions $\mathbf{E}_1$ and $\mathbf{E}_2$ for a torus domain, where $k=3.1782$. The top row and bottom row denote the three components of eigenfunctions of $\mathbf{E}_1$ and $\mathbf{E}_2$ respectively.}
\end{figure}

\begin{example}
In this example, we  consider a constant refractive index case and set $\varepsilon=64$. We first choose a unit ball as the domain $\Omega$. The mesh is generated by the open-source soft GMSH \cite{Geuzaine09} and  the number of the degree of freedom is $\mathrm{DoF}=23089$.
We use an interval $[10, 12]$ to search for transmission eigenvalues and the corresponding eigenfunctions.  Figure \ref{fig:ball-mono} shows the slice plots of transmission eigenfunctions associated with transmission eigenvalue $k=3.3318$.  It is clear that  the eigenfunctions are surface-localized around the boundary for  $\mathbf{E}_2$ but not surface-localized around the boundary for $\mathbf{E}_1$, which fulfil the mono-localized transmission eigenmodes.
Moreover, Fig.~\ref{fig:ball-bi} presents a pair of transmission eigenfunctions that satisfy the bi-localized transmission eigenmodes with $k=3.2296$. One can find that the localized behavior is obvious for $\mathbf{E}_2$ while not obvious for $\mathbf{E}_1$, which is because the eigenvalue $k$ is not sufficiently large. In fact, the transmission eigenfunctions of $\mathbf{E}_1$ almost vanish in the centre area of the unit ball and are localized near the boundary. In addition, we must emphasize that bi-localized cases are much fewer than mono-localized cases in numerics.

We now consider a more general domain, i.e., torus domain. The major radius is $1.5$ and the minor radius is $0.5$.  In order to resolve the model, we introduce finer meshes and the number of the degree of freedom for torus is $\mathrm{DoF}=52108$.  Interval $[11,12]$ is used for searching transmission eigenvalues and eigenfunctions.  All computational results only fulfil the mono-localized transmission eigenmodes as shown in Fig.~\ref{fig:torus-mono}.
\end{example}

\begin{example}
\begin{figure}
    \subfigure[$E_2^{(1)}$]{\includegraphics[width=0.32\textwidth]{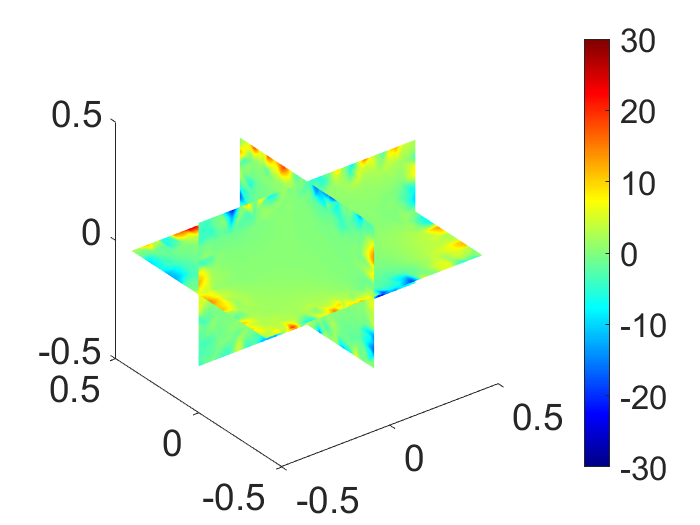}}
    \subfigure[$E_2^{(2)}$]{\includegraphics[width=0.32\textwidth]{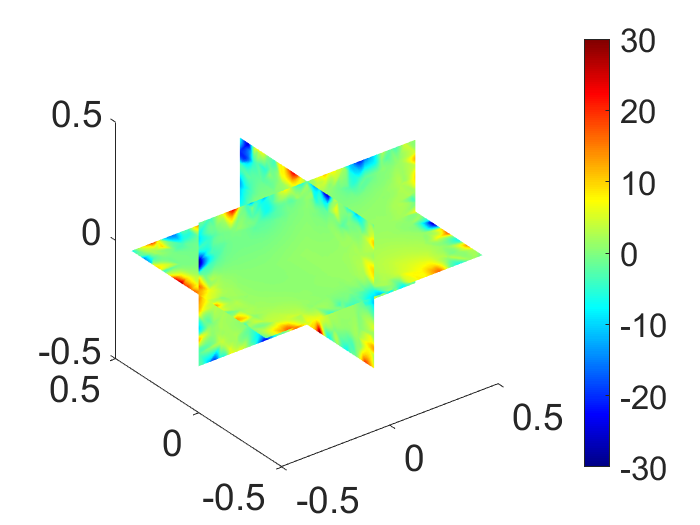}}
    \subfigure[$E_2^{(3)}$]{\includegraphics[width=0.32\textwidth]{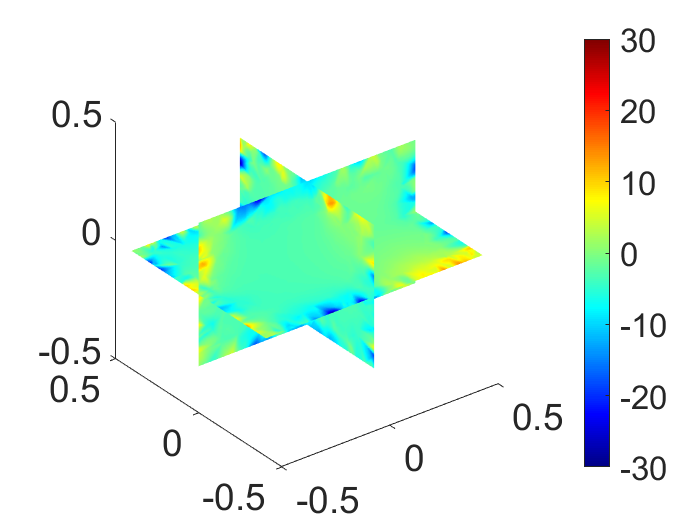}}\\
    \subfigure[plane $x=0$]{\includegraphics[width=0.32\textwidth]{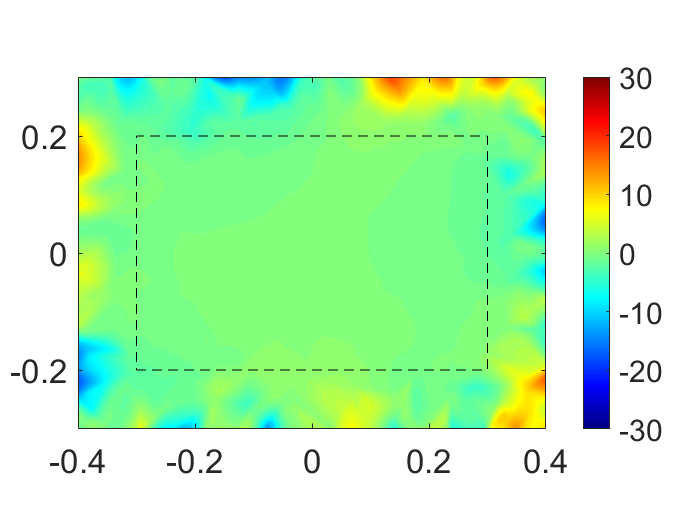}}
    \subfigure[plane $y=0$]{\includegraphics[width=0.32\textwidth]{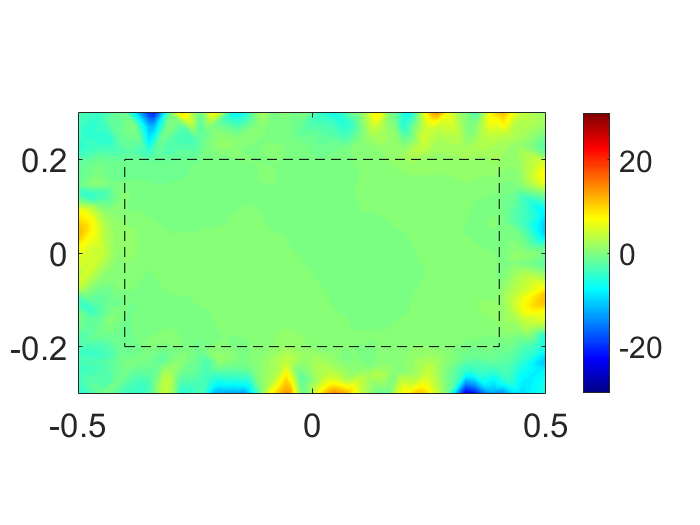}}
    \subfigure[plane $z=0$]{\includegraphics[width=0.32\textwidth]{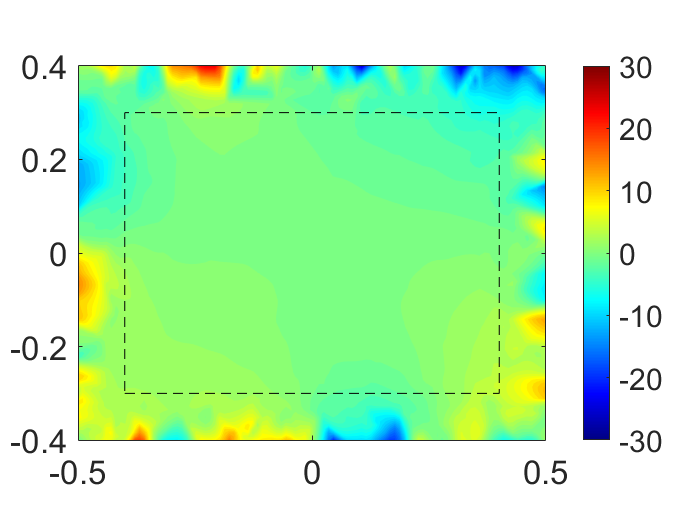}}\\
    \caption{\label{fig:cub-mono} Mono-localized transmission eigenmodes. Slice plots of eigenfunctions $\mathbf{E}_2$ for a cuboid domain, where $k=3.3448$. The top row denotes the three components of eigenfunctions of $\mathbf{E}_2$ and the bottom row denotes slice plots of $\mathbf{E}_2^{(1)}$ at three planes.}
\end{figure}
In this example, we consider a case with a piecewise constant refractive index. We construct a two-layer cuboid domain centered at the origin, where the edge lengths of the inner cuboid are $(0.8, 0.6, 0.4)$ and $(1, 0.8, 0.6)$ for the outer cuboid. Here, we set $\varepsilon=4$ inside the inner cuboid and $\varepsilon=400$ between the inner and outer cuboids. To provide a more accurate solution with less computational cost,  we establish a relatively rough mesh in the inner cuboid but a fine mesh outside the inner cuboid and the number of the degree of freedom is DoF $= 65394$. Fig. \ref{fig:cub-mono} presents the transmission eigenfunctions $\mathbf{E}_2$  for the mono-localized transmission eigenmodes with $k=3.3448$.
In the figures, the dotted black lines denote the contour of the inner cuboid, which is the interface of different refractive indices. One can observe that the transmission eigenfunctions are actually surface-localized around the boundary even if the refractive index is piecewise-constant, see Fig. \ref{fig:cub-mono} (d)(e)(f).


\end{example}

\section{Topological properties of Maxwell transmission eigenfunctions}

In what follows, with a bit abuse of notations, we let $\varepsilon, \mu$ be symmetric-positive-definite-matrix valued functions in $\Omega$, and consider the Maxwell system
\begin{equation}\label{eq:m1}
\nabla\wedge\mathbf{E}=\mathrm{i}k\mu\mathbf{H},\quad\nabla\wedge\mathbf{H}=-\mathrm{i}k\varepsilon\mathbf{E}\quad\mbox{in}\ \ \Omega.
\end{equation}
We recall the following transformation-invariant property of the Maxwell system.
\begin{lem}[Lemma 2.2 in \cite{LZ}]\label{lem:t1}
Consider a bi-Lipschitz and orientation-preserving coordinate transformation $\widetilde{\mathbf{x}}=\mathbf{F}(\mathbf{x}):\Omega\rightarrow\widetilde{\Omega}$. Let $D\mathbf{F}$ denote the Jacobian matrix of $\mathbf{F}$. Assume that $(\mathbf{E},\mathbf{H})\in H(\mathrm{curl}, \Omega)\times H(\mathrm{curl}, \Omega)$ are EM fields to \eqref{eq:m1}; then for the pull-back fields given by
\begin{equation*}\label{eq:p1}
\widetilde{\mathbf{Q}}=(\mathbf{F}^{-1})^*\mathbf{Q}:=(D\mathbf{F}^T)\mathbf{Q}\circ\mathbf{F}^{-1},\quad \mathbf{Q}=\mathbf{E}, \mathbf{H},
\end{equation*}
we have $(\mathbf{E}, \mathbf{H})\in H(\mathrm{curl}, \Omega)\times H(\mathrm{curl}, \Omega)$, which satisfy the Maxwell system
\begin{equation*}\label{eq:p2}
\widetilde{\nabla}\wedge\widetilde{\mathbf{E}}=\mathrm{i}k\tilde{\mu}\widetilde{\mathbf{H}},\quad \widetilde{\nabla}\wedge\widetilde{\mathbf{H}}=-\mathrm{i}k\widetilde{\varepsilon}\widetilde{\mathbf{E}},
\end{equation*}
where $\widetilde{\nabla}$ signifies the differentiation in the $\widetilde{\mathbf{x}}$-coordinate, and $\tilde{\varepsilon}$ and $\tilde{\mu}$ are the push-forwards of $\varepsilon$ and $\mu$ via $\mathbf{F}$, defined by
\begin{equation*}\label{eq:pf1}
\tilde{\zeta}=\mathbf{F}_*\zeta:=|D\mathbf{F}|^{-1} D\mathbf{F}\cdot\zeta\cdot D\mathbf{F}^T\circ\mathbf{F}^{-1},\quad\zeta:=\varepsilon, \mu.
\end{equation*}
\end{lem}

Using Lemma~\ref{lem:t1}, one can readily obtain the following transformation property of the transmission eigenvalue problem.

\begin{thm}\label{thm:t1}
Let $\Omega, \widetilde{\Omega}$ and $\mathbf{F}$ be the same as those in Lemma~\ref{lem:t1}. Assume that $\mathbf{Q}_j\in H(\mathrm{curl}, \Omega)$, $j=1, 2$ and $\mathbf{Q}=\mathbf{E}, \mathbf{H}$, satisfy the following transmission eigenvalue problem:
\begin{equation}\label{eq:tt1}
\left\{
\begin{aligned}
&\nabla\wedge\mathbf{E}_j-\mathrm{i}k\mu_j\mathbf{H}_j=\mathbf{0}&&\mbox{in}\ \Omega,\\
&\nabla\wedge\mathbf{H}_j+\mathrm{i}k\varepsilon_j\mathbf{E}_j=\mathbf{0} &&\mbox{in}\ \Omega,\\
&\nu\wedge\mathbf{E}_j=\nu\wedge\mathbf{H}_j && \mbox{on}\ \partial\Omega, \ \ j=1, 2,
\end{aligned}
\right.
\end{equation}
the the pull-back fields $\widetilde{\mathbf{Q}}_j=(\mathbf{F}^{-1})^*\mathbf{Q}_j$ are solutions to the following transmission eigenvalue problem:
\begin{equation}\label{eq:tt2}
\left\{
\begin{aligned}
&\nabla\wedge\widetilde{\mathbf{E}}_j-\mathrm{i}k\tilde\mu_j\widetilde{\mathbf{H}}_j=\mathbf{0}, && \mbox{in}\ \widetilde\Omega,\\
&\nabla\wedge\widetilde{\mathbf{H}}_j+\mathrm{i}k\tilde\varepsilon_j\widetilde{\mathbf{E}}_j=\mathbf{0}&& \mbox{in}\ \widetilde\Omega,\\
&\nu\wedge\widetilde{\mathbf{E}}_j=\nu\wedge\widetilde{\mathbf{H}}_j && \mbox{on}\ \partial\widetilde\Omega, \ \ j=1, 2,
\end{aligned}
\right.
\end{equation}
where $(\tilde\varepsilon_j, \tilde\mu_j)=\mathbf{F}_*(\varepsilon_j, \mu_j)$, $j=1, 2$.
\end{thm}
\begin{proof}
Using Lemma~\ref{lem:t1}, one can readily have that the EM fields in \eqref{eq:tt1} and \eqref{eq:tt2} are related by the pull-back while the material parameters are related by the push-forward through the transformation $\mathbf{F}$. Hence, it is sufficient to show that the transmission conditions on $\partial\Omega$ and $\partial\widetilde\Omega$ are invariant under the transformation, which can be straightforwardly shown by following the argument in \cite{LZ} in proving Lemma~\ref{lem:t1}.
\end{proof}

Let us consider a specific case by letting $\widetilde\Omega$ be obtained through a translation of $\Omega$, namely $\widetilde\Omega=\mathbf{z}+\Omega$ with $\mathbf{z}\in\mathbb{R}^3$. In such a case, we call $\widetilde\Omega$ a copy of $\Omega$. Suppose $\mathbf{E}_j(\mathbf{x})$, $\mathbf{x}\in\Omega$ and $j=1, 2$, are transmission eigenfunctions to \eqref{eq:tcase12} associated with $(\Omega, \varepsilon)$. Using Theorem~\ref{thm:t1}, it is straightforward to verify that $\widetilde{\mathbf{E}}_j(\mathbf{x}-\mathbf{z})$, $\mathbf{x}\in\widetilde\Omega$ and $j=1, 2$, are transmission eigenfunctions associated with $(\widetilde\Omega, \tilde\varepsilon)$, where $\tilde\varepsilon(\mathbf{x})=\varepsilon(\mathbf{x}-\mathbf{z})$ for $\mathbf x\in\widetilde\Omega$. This fact shall be used in our subsequent study of the artificial electromagnetic mirage.

Next, we consider the case that the underlying domain $\Omega$ has multiply connected components. {To that end, we first recall that for a fixed transmission eigenvalue $k\in\mathbb{R}_+$ in \eqref{eq:tcase12}, the corresponding eigenspace is finite-dimensional \cite{CHreview}.} In what follows, we let $\mathcal{S}_{k, (\Omega, \varepsilon)}$ signify the eigenspace associated with the transmission eigenvalue $k$ and the medium $(\Omega, \varepsilon)$ in \eqref{eq:tcase12}.

\begin{thm}\label{thm:t2}
Consider the transmission eigenvalue problem \eqref{eq:tcase12} and assume that $(\Omega, \varepsilon)=\cup_{j=1}^l (\Omega_j, \varepsilon_j)$, where each $\Omega_j$ is a simply-connected component of $\Omega$. Then it holds that
\begin{equation*}\label{eq:t21}
\mathrm{dim}(\mathcal{S}_{k, (\Omega_j, \varepsilon_j)})=p_j \quad\mbox{and}\quad \mathcal{S}_{k, (\Omega_j, \varepsilon_j)}=\mathrm{Span}\{(\mathbf{E}_1^{\alpha_j}, \mathbf{E}_2^{\alpha_j})_{\alpha_j=1}^{p_j}\}
\end{equation*}
if and only if
\begin{equation*}\label{eq:t22}
\mathrm{dim}(\mathcal{S}_{k, (\Omega, \varepsilon)})=\sum_{j=1}^{l} p_j:=p\quad\mbox{and}\quad \mathcal{S}_{k, (\Omega, \varepsilon)}=\mathrm{Span}\{(\mathbf{E}_1^\alpha, \mathbf{E}_2^\alpha)_{\alpha=1}^p\},
\end{equation*}
where $(\mathbf{E}_1^\alpha, \mathbf{E}_2^\alpha)$ is of the form
\begin{equation*}\label{eq:t23}
(\mathbf{E}_1^\alpha, \mathbf{E}_2^\alpha)=(\mathbf{E}_1^{\alpha_j}, \mathbf{E}_2^{\alpha_j})\chi_{\Omega_j}+(\mathbf{0}, \mathbf{0})\chi_{\cup_{\beta=1, \beta\neq j}^l\Omega_\beta},
\end{equation*}
for a certain $1\leq j\leq l$.

\end{thm}

\begin{proof}

First, it is straightforward to verify that $k$ is a transmission eigenvalue associated with $(\Omega, \varepsilon)$ if and only if $k$ is a transmission eigenvalue associated with each $(\Omega_j, \varepsilon_j)$ for $j=1, 2, \ldots, l$.

Next, by using \eqref{eq:tcase12}, it is directly verified that the non-trivial functions defined in \eqref{eq:t23} are a pair of transmission eigenfunctions associated with $(\Omega, \varepsilon)$. On the other hand, if $(\mathbf{E}_1, \mathbf{E}_2)$ are transmission eigenfunctions to \eqref{eq:tcase12}, it can be shown that $(\mathbf{E}_1|_{\Omega_j}, \mathbf{E}_2|_{\Omega_j})$ are transmission eigenfunctions associated with $(\Omega_j, \varepsilon_j)$. Hence, one has
\[
(\mathbf{E}_1|_{\Omega_j}, \mathbf{E}_2|_{\Omega_j})\in \mathcal{S}_{k, (\Omega_j, \varepsilon_j)}=\mathrm{Span}\{(\mathbf{E}_1^{\alpha_j}, \mathbf{E}_2^{\alpha_j})_{\alpha_j=1}^{p_j}\},
\]
which readily implies that $\mathcal{S}_{k, (\Omega, \varepsilon)}$ is spanned by the $p$ functions defined in \eqref{eq:t23}.

The proof is complete.

\end{proof}

Consider the special example again as above by letting $\Omega$ be simply-connected and $\widetilde\Omega$ be its copy. In order to further simplify the exposition, we let $\varepsilon$ be constant. According to Theorem~\ref{thm:t2}, the multiplicity of the eigenvalue $k$ associated with $(\Omega\cup\widetilde\Omega, \varepsilon)$ doubles that of $(\Omega, \varepsilon)$. Moreover, the corresponding eigenfunctions are given by
\begin{equation}\label{eq:c1}
(\mathbf{E}_1, \mathbf{E}_2)\chi_{\Omega}+(\mathbf{0}, \mathbf{0})\chi_{\widetilde\Omega},\quad (\mathbf{0}, \mathbf{0})\chi_{\Omega}+(\widetilde{\mathbf{E}}_1, \widetilde{\mathbf{E}}_2)\chi_{\widetilde\Omega},
\end{equation}
as well as the corresponding linear combination of the above two types of resonant modes. Here, $(\mathbf{E}_1,\mathbf{E}_2)$ are the resonant modes associated with $(\Omega, \varepsilon)$ and $(\widetilde{\mathbf{E}}_1,\widetilde{\mathbf{E}}_2)$ are the resonant modes associated with $(\widetilde{\Omega}, \varepsilon)$, and by Theorem~\ref{thm:t1}, $\widetilde{\mathbf{E}}_j(\mathbf{x})=\mathbf{E}_j(\mathbf{x}-\mathbf{z})$, $\mathbf{x}\in\widetilde\Omega$. This fact shall also be used in our subsequent study of the artificial mirage.

Finally, we would like to point out that Theorem~\ref{thm:t2} can easily extended to the case with more general transmission eigenvalue problems as considered in Theorem~\ref{thm:t1}. However, we shall not explore this point further.

\section{Artificial electromagnetic mirage}

In this section, we propose a novel and interesting application of the geometrical and topological properties establish in the previous sections for the transmission resonances.
The practical scenario can be described as follows.

Let $(\Omega, \varepsilon)$ be an optical object, where $\Omega$ is simply-connected and $\varepsilon$ is  constant. Our goal is to generate certain incident electromagnetic field $(\mathbf{E}^i, \mathbf{H}^i)$ impinging upon $(\Omega,\varepsilon)$ with the following properties:
\begin{itemize}
\item $(\mathbf{E}^i, \mathbf{H}^i)$ is a solution to \eqref{eq:in1};
\item the propagation of the EM fields is nearly non-interrupted by the presence of the physical object $(\Omega, \varepsilon)$;
\item the total wave field is localized on the boundary of a predefined set $\widetilde{\Omega}\subset\mathbb{R}^3\backslash\bar{\Omega}$, where $\widetilde{\Omega} = \Omega + \mathbf{z}$ is simply a translation of $\Omega$. This copy is referred to as an artificial mirage image of the optical object $(\Omega, \varepsilon)$.
\end{itemize}

Now we introduce the scheme to get the desired $(\mathbf{E}^i, \mathbf{H}^i)$ given $\Omega, \varepsilon$ and $\mathbf{z}$. First, we recall the Fourier extension of the EM fields, which is also known as the Herglotz approximation. For a tangential vector field $\mathbf{g}\in L^2(\mathbb{S}^2)^3$ on the unit sphere, we define
\begin{equation}\label{eq:h1}
\mathbf{E}_{\mathbf{g}}(\mathbf{x}):=\int_{\mathbb{S}^2} e^{\mathrm{i}k\mathbf{x}\cdot \mathbf{d}} \mathbf{g}(\mathbf{d})\, ds(\mathbf{d}),\quad \mathbf{H}_{\mathbf{g}}(\mathbf{x}):=\frac{1}{\mathrm{i}k}\nabla\wedge\mathbf{E}_{\mathbf{g}}(\mathbf{x})
\end{equation}
to be the Herglotz fields associated with $\mathbf{g}$. For any regular domain $\Sigma$ with a connected complement in $\mathbb{R}^3$, it is known that the solution $(\mathbf{E}, \mathbf{H})\in H(\mathrm{curl}, \Sigma)\times H(\mathrm{curl}, \Sigma)$ to the following system
\begin{equation}\label{eq:ee1}
\nabla\wedge\mathbf{E} -\mathrm{i}k\mathbf{H}=\mathbf{0},\quad \nabla\wedge\mathbf{H}+\mathrm{i}k\mathbf{E}=\mathbf{0}\quad\mbox{in}\ \Sigma,
\end{equation}
can be approximated by the Herglotz fields to an arbitrary accuracy; see \cite{Wec}. It is noted that the Herglotz approximation result in \cite{Wec} requires a bit stronger regularities on the EM fields in \eqref{eq:ee1}, which we shall always assume in our subsequent discussion.

Recall that $(\widetilde\Omega, \varepsilon)$ is an exact copy of $(\Omega, \varepsilon)$, where $\widetilde\Omega:=\mathbf{z}+\Omega$ and $\mathbf{z}\in\mathbb{R}^3$.  Consider the transmission eigenvalue problem \eqref{eq:tcase12} associated with $(\Omega, \varepsilon)\cup (\widetilde\Omega, \varepsilon)$. By Theorem~\ref{thm:t2},  and in particular the example discussed after the theorem, we know that the corresponding transmission eigenfunctions are given by \eqref{eq:c1}. By the geometrical results in Theorem \ref{thm:main} in Section 2, we know that there exist infinitely many eigenfunctions associated with $(\Omega, \varepsilon)$ which are surface-localized around $\partial\Omega$. Hence, we can choose one of such resonant modes such that $\mathbf{E}_2$ is surface-localized around $\partial\Omega$ and so $\widetilde{\mathbf{E}}_2$ is surface-localized around $\partial\widetilde{\Omega}$, (see Fig. \ref{fig:mirage} for a schematic illustration).
\begin{figure}
 {\includegraphics[width=0.7\textwidth]{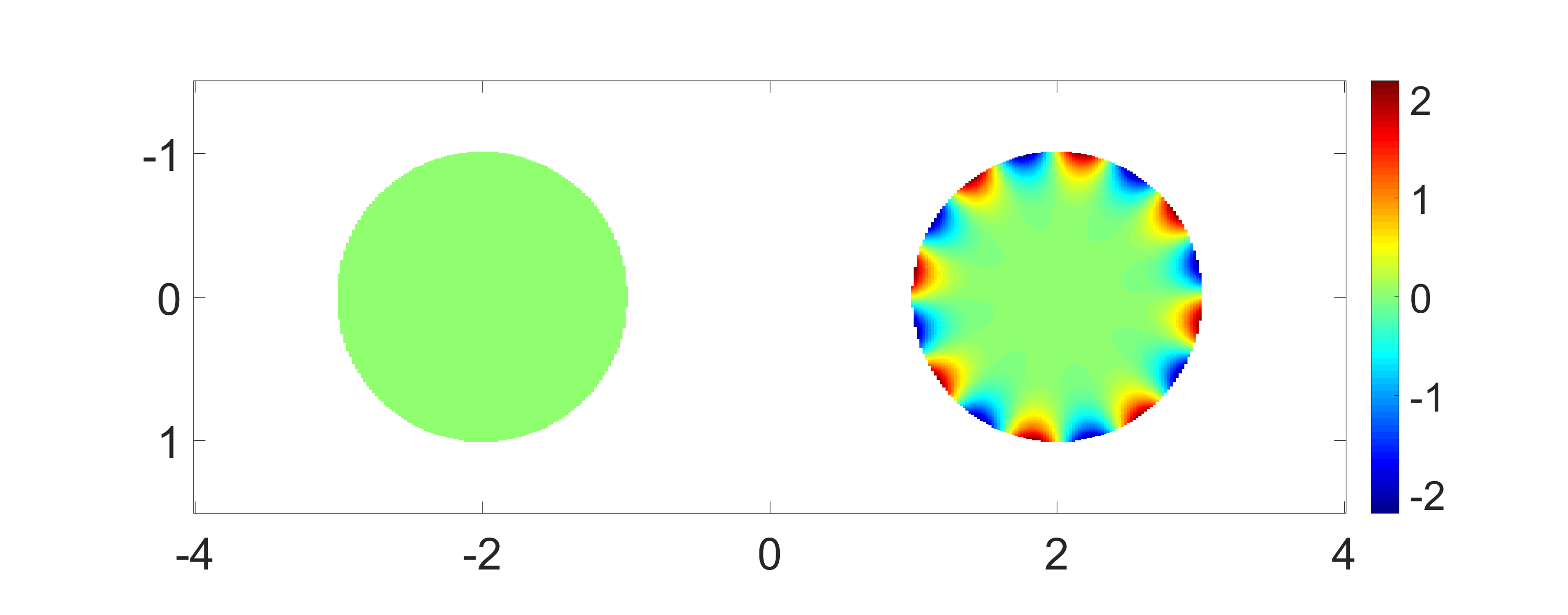}}
    \caption{\label{fig:mirage} Contour plot of the transmission eigenfunction that is vanished on the left domain but localized on the right domain.}
\end{figure}
Next, we consider the second resonant mode in \eqref{eq:c1} and let ${\mathbf{E}}_{\mathbf{g}}^i$ be the Herglotz field associated with $\mathbf{0}\cdot\chi_\Omega+\widetilde{\mathbf{E}}_2\cdot\chi_{\widetilde\Omega}$, i.e.
\begin{equation}\label{eq:h2}
\mathbf{E}^i_{\mathbf{g}}\approx \mathbf{0}\cdot\chi_\Omega+\widetilde{\mathbf{E}}_2\cdot\chi_{\widetilde\Omega}\quad\mbox{in}\ \ \Omega\cup\widetilde\Omega.
\end{equation}
Let $\mathbf{H}_{\mathbf{g}}^i$ be given according to the second formula in \eqref{eq:h1} associated with $\mathbf{E}_{\mathbf{g}}^i$. Clearly, $(\mathbf{E}_{\mathbf{g}}^i, \mathbf{H}_{\mathbf{g}}^i)$ are entire solutions to \eqref{eq:in1}, and can be used as incident fields, namely
\begin{equation}\label{eq:h2.5}
(\mathbf{E}^i, \mathbf{H}^i)=(\mathbf{E}_{\mathbf{g}}^i, \mathbf{H}_{\mathbf{g}}^i).
\end{equation}
Next, we consider the electromagnetic scattering problem \eqref{eq:maxwell1} associated with the optical object $(\Omega, \varepsilon)$ and incident field $(\mathbf{E}^i,\mathbf{H}^i)$. It is straightforward to deduce from \eqref{eq:in1} and \eqref{eq:maxwell1} that the corresponding scattering field $(\mathbf{E}^s, \mathbf{H}^s)$ satisfies
\begin{equation}\label{eq:maxwell12}
\left\{
\begin{aligned}
&~~\nabla\wedge\mathbf{E}^s-\mathrm{i}k\mathbf{H}^s=\mathbf{0} && \mbox{in}\ \mathbb{R}^3,\\
&~~\nabla\wedge\mathbf{H}^s+\mathrm{i}k\varepsilon\mathbf{E}^s=\mathrm{i}k(1-\varepsilon)\mathbf{E}^i  &&\mbox{in}\ \mathbb{R}^3,\\
&\lim_{|{\mathbf x}| \rightarrow \infty}(\mathbf{H}^{s} \wedge {\mathbf x}-|{\mathbf x} | \mathbf{E}^{s})={\mathbf 0}.
\end{aligned}
\right.
\end{equation}
By \eqref{eq:h2}, we see that
\begin{equation*}\label{eq:h3}
\mathrm{i}k(1-\varepsilon)\mathbf{E}^i\approx \mathbf{0}\ \ \mbox{in}\ \mathbb{R}^3.
\end{equation*}
Therefore, by the well-posedness of the forward scattering system \eqref{eq:maxwell12}, we readily conclude that
\begin{equation*}\label{eq:h4}
\mathbf{E}^s, \mathbf{H}^s\approx \mathbf{0}\quad\mbox{in}\ \ \mathbb{R}^3\backslash\overline{\Omega}.
\end{equation*}
That is, for the total fields $\mathbf{E}$ and $\mathbf{H}$, one has
\begin{equation}\label{eq:h5}
(\mathbf{E}, \mathbf{H})\approx (\mathbf{E}_{\mathbf{g}}^i, \mathbf{H}_{\mathbf{g}}^i)\quad\mbox{in}\ \ \mathbb{R}^3\backslash\overline{\Omega}.
\end{equation}
Finally, by \eqref{eq:h2} and \eqref{eq:h5}, we clearly have that the total electric field $\mathbf{E}$ is localized around $\partial\widetilde\Omega$. That is, a ``shining" area, i.e. $\partial\widetilde\Omega$, is produced and the electromagnetic fields inside $\widetilde\Omega$ are nearly zero.

To summarize, given an optical object $(\Omega, \varepsilon)$, if one uses the incident fields designed in \eqref{eq:h2} and \eqref{eq:h2.5} to illuminate  the object, an artificial mirage copy $(\widetilde\Omega, \varepsilon)$ will be produced. In what follows, we shall provide a few examples to illustrate the artificial mirage effect. However, in order to obtain eminent surface-localizing behaviour of the transmission eigenfunctions, one needs to compute the transmission eigenfunctions associated to \eqref{eq:tcase12} with either very large $\varepsilon$ or very large $k$. This requires extensive and heavy costs in numerically computing the three-dimensional transmission eigenvalue problem \eqref{eq:tcase12} as well as the corresponding scattering system \eqref{eq:maxwell1}. Due to limited computational resources, we could only afford to conduct the two-dimensional simulations, which correspond to the TM (transverse magnetic) scattering.

\begin{figure}
    \subfigure[]{\includegraphics[width=0.48\textwidth]{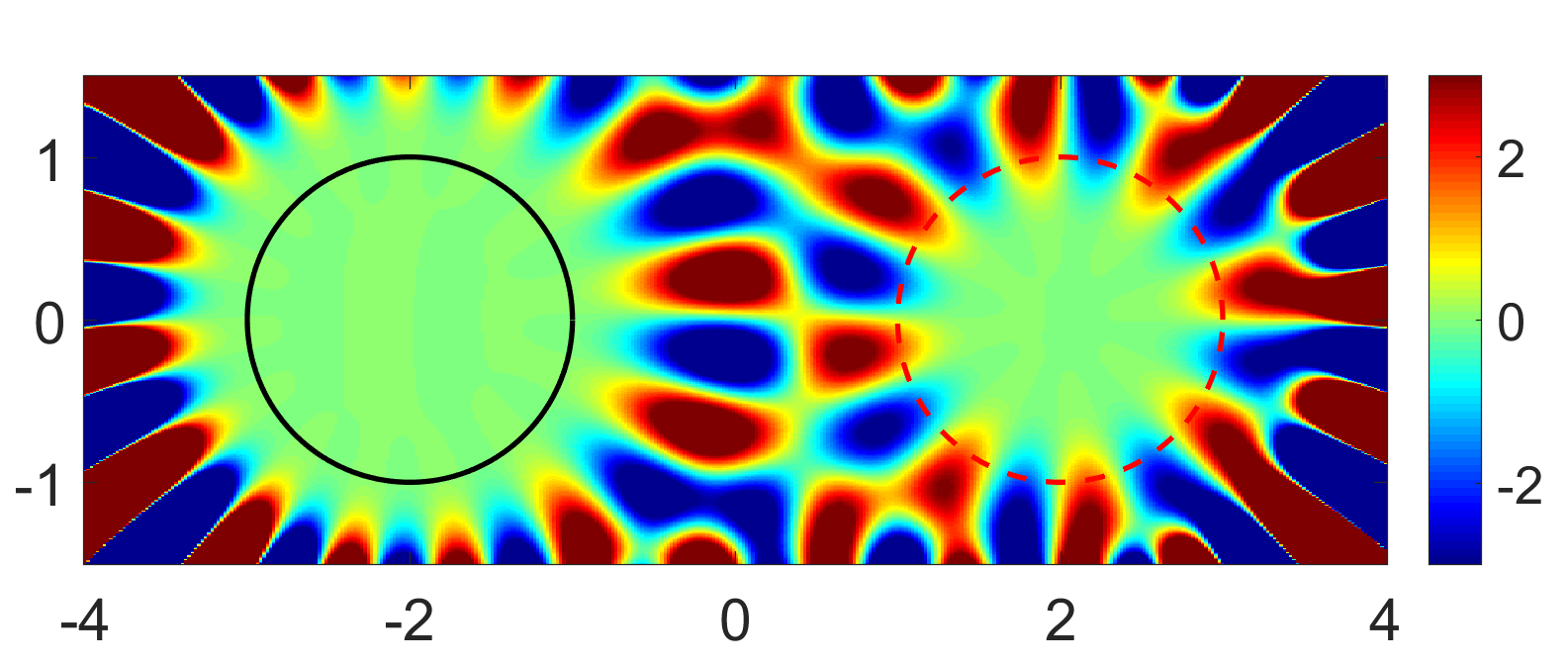}}
    \subfigure[]{\includegraphics[width=0.48\textwidth]{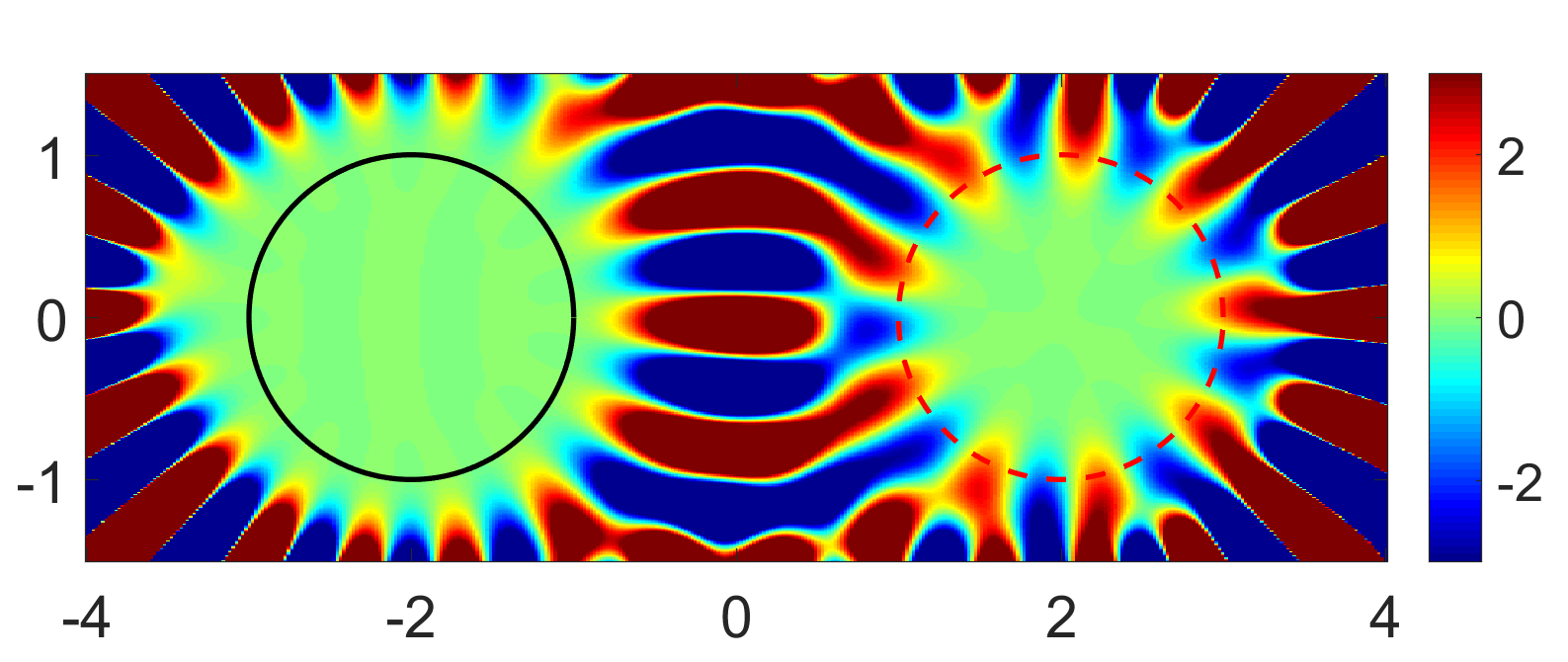}}\\
    \subfigure[]{\includegraphics[width=0.48\textwidth]{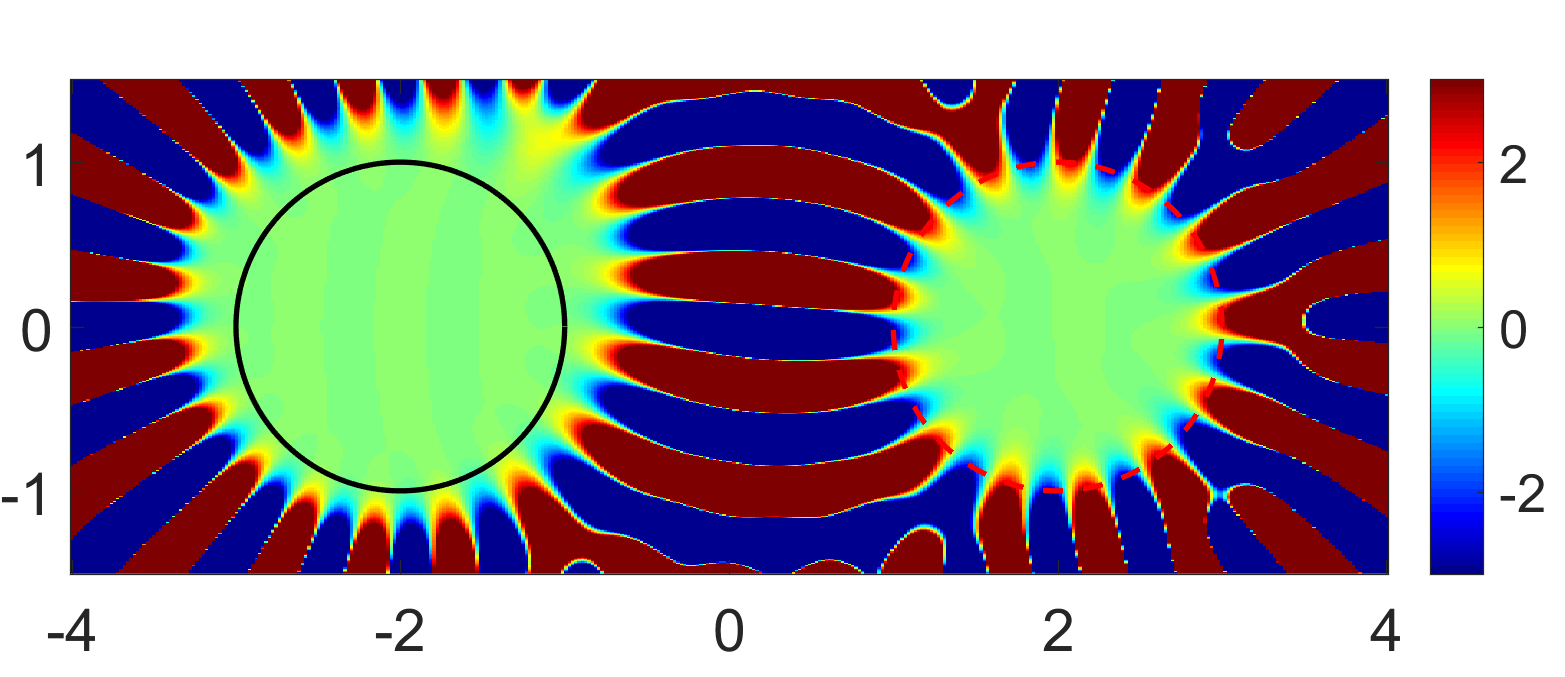}}
    \subfigure[]{\includegraphics[width=0.48\textwidth]{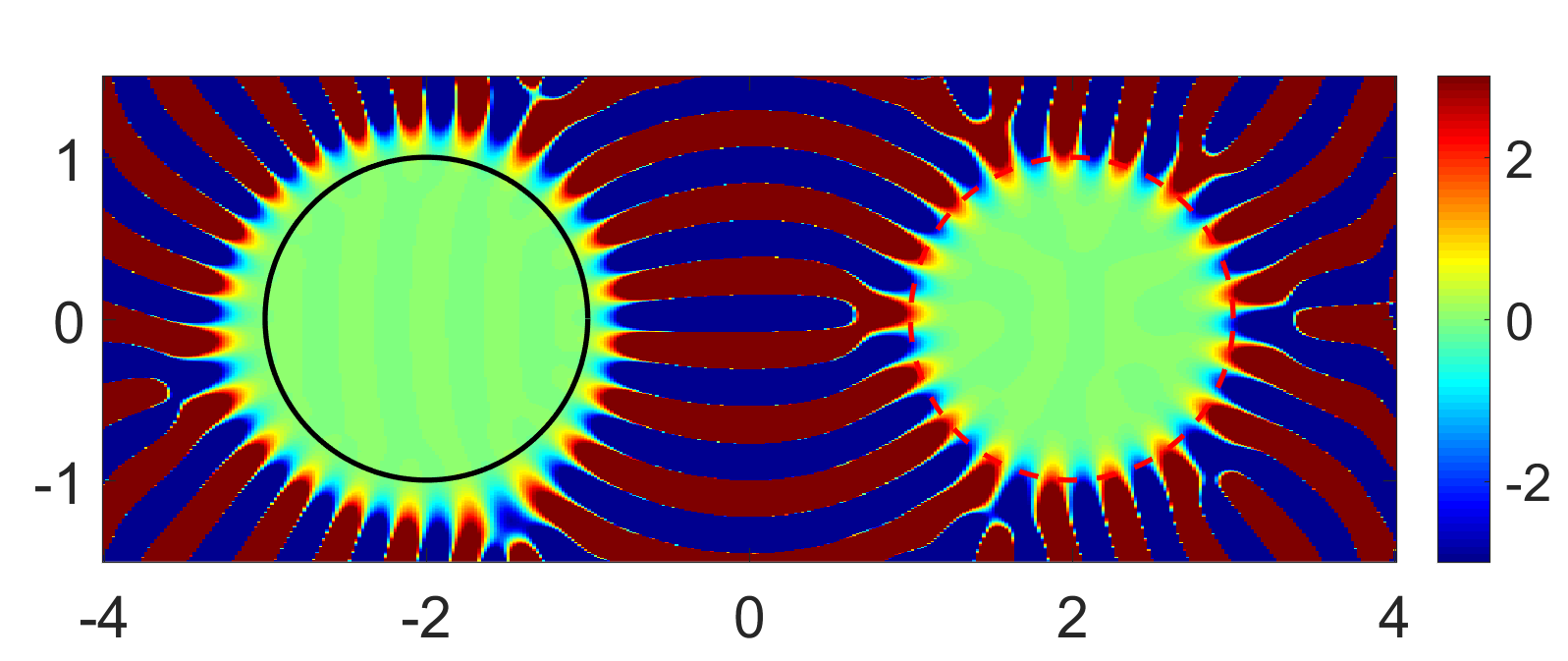}}
    \caption{\label{fig:two-circle1} Contour plots of the wave field for a disk with different frequencies. (a)  $k=7.38$, (b) $k=8.66$, (c) $k=10.08$, (d) $k=14.09$.}
\end{figure}

\begin{example}
In this example, we present a mirage phenomenon for a unit disk centering at $(-2,0)$.
 The electric permittivity is given by $\epsilon_0=4$ in the disk and $\epsilon_0=1$ outside the disk.
 Figure \ref{fig:two-circle1} shows total field with different  frequencies, i.e., transmission eigenvalues, where the solid black line denotes the exact disk and the dotted red denotes the virtual disk.
  One can easily observe that the wave field is surface-localized around the boundary virtual disk, and the virtual object could be easily identified. By comparing the results in Figure \ref{fig:two-circle1}, we can observe that the total field has a better localized result when the wavenumber is larger.

\end{example}

\begin{example}
In this example, we are concerned with the mirage phenomenon for a more general case.
Let $\Omega$ be the triangle shaped domain which is described by the parametric representation
\begin{equation*}
  \mathbf{x}(t)=0.5(2+0.2\cos 3t)(\cos t, \ \sin t) \quad t\in[0,\, 2\pi],
\end{equation*}
 where the center is given by $(-2,0)$ and permittivity $\epsilon_0=4$ for $\mathbf{x}\in \Omega$.
 Figure \ref{fig:two-triangles} presents the wave field  around the exact domain (solid black line) and  virtual domain (dotted red line). It demonstrates very good imaging performance of  the mirage for general convex domain.

\begin{figure}
    \subfigure[]{\includegraphics[width=0.48\textwidth]{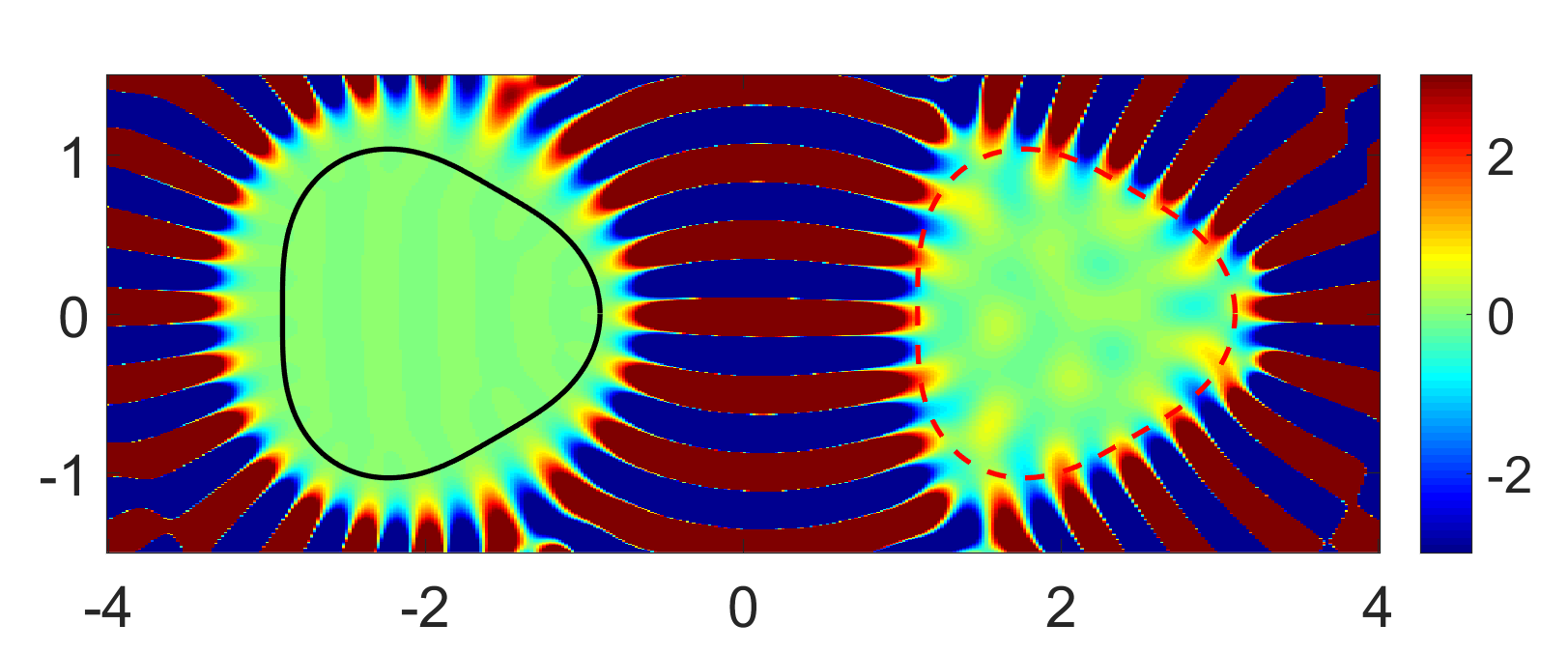}}
    \subfigure[]{\includegraphics[width=0.48\textwidth]{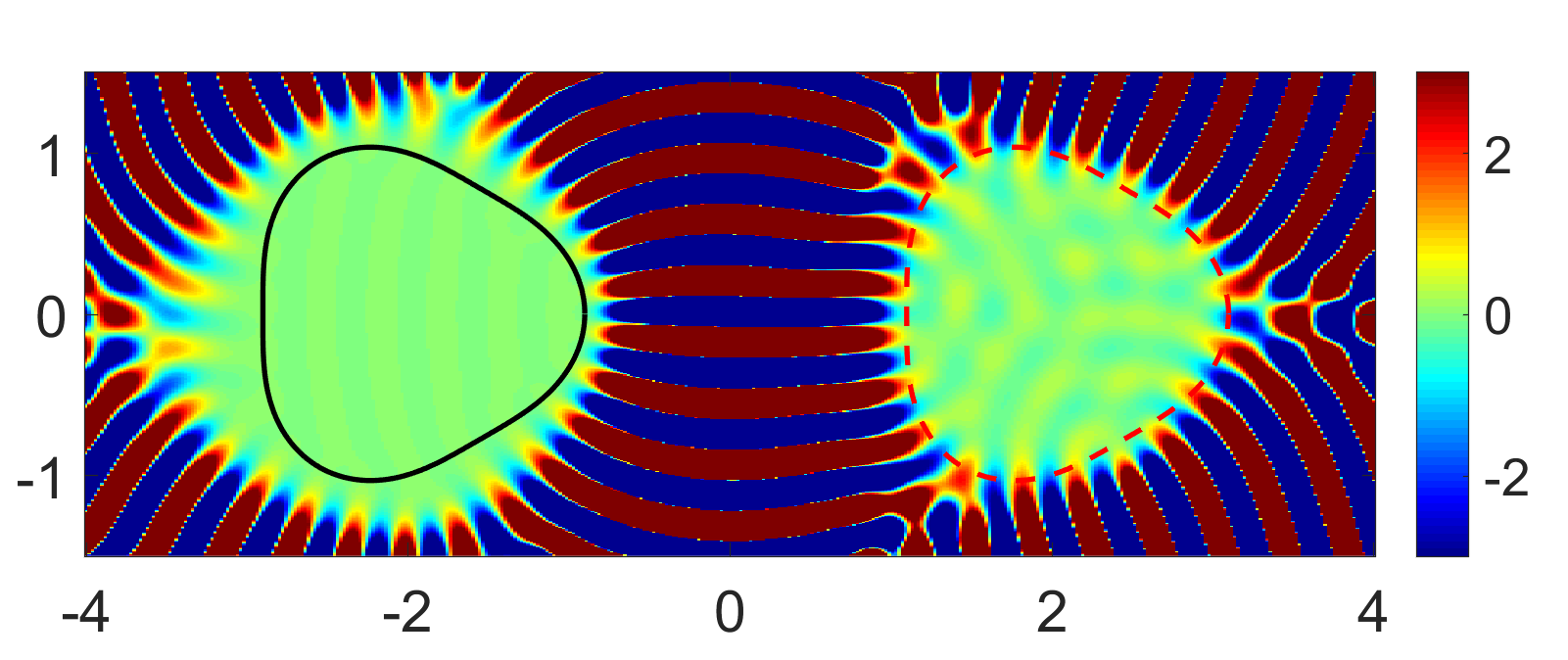}}
    \caption{\label{fig:two-triangles} Contour plots of the wave field for a triangle-shaped domain with different frequencies. (a) $k=13.25$, (b) $k=16.84$.} 
\end{figure}

\end{example}

\section*{Acknowledgment}
The work of Y. Deng was supported by NSF grant of China No. 11971487 and NSF grant
of Hunan No. 2020JJ2038. The work of H. Liu was supported by a startup grant from City
University of Hong Kong and Hong Kong RGC General Research Funds (projects 12301218,
12302919 and 12301420). The work of X. Wang was supported by the Hong Kong Scholars
Program grant under No. XJ2019005 and NSF grant of China(No. 11971133 and No. 12001140). The work of W. Wu was supported by the tier-2 startup grant from Hong Kong Baptist University and Hong Kong RGC General Research Funds (projects 12302219 and 12300520).


\end{document}